\documentclass[a4paper,10pt]{article}
\usepackage[dvipdfmx]{hyperref}
\usepackage{url,graphicx}
\usepackage{amsthm,amsmath,amssymb,amsfonts}
\usepackage{showexpl}
\usepackage{mathrsfs,xr}
\usepackage[all,arc,curve,color,frame]{xy}
\usepackage{graphics,fullpage}
\xyoption{all}

\title{The JLO Character for
The Noncommutative Space of Connections of 
Aastrup-Grimstrup-Nest}

\author{Alan Lai
\thanks{Email: alan@math.toronto.edu}
\\
University of Toronto
}

\begin{document}
\maketitle

\newcommand{\cch}
{\widetilde{\mathtt{Ch}}} 
\newcommand{\Dd}{\mathcal{D}}
\newcommand{\sDd}{\slash\hspace{-0.26cm}\mathcal{D}}
\newcommand{\dDd}{\Dd_I}
\newcommand{\CCc}{\overline{\Cc}^{\Ggamma}}
\newcommand{\TR}{\operatorname{Tr}}
\newcommand{\CCHERN}{\operatorname{ch}}
\newcommand{\KKEI}{\operatorname{K}}
\newcommand{\g}{\mathfrak{g}} 
\newcommand{\naturalnumber}{\mathbb{N}}
\newcommand{\CAS}{\operatorname{Cas}}
\newcommand{\MAP}{\operatorname{Map}}
\newcommand{\GAU}{\operatorname{Gau}}
\newcommand{\MATRIX}{\operatorname{Mat_\mathbb{C}(N)}}
\newcommand{\END}{\operatorname{End}}
\newcommand{\Uu}{\mathcal{U}}
\newcommand{\Ww}{\mathcal{W}}
\newcommand{\Aa}{\mathcal{A}}
\newcommand{\Bb}{\mathcal{B}}
\newcommand{\Hh}{\mathcal{H}}
\newcommand{\Ss}{\mathcal{S}}
\newcommand{\Ff}{\mathfrak{F}}
\newcommand{\Nn}{\mathcal{N}}
\newcommand{\Pp}{\mathcal{P}}
\newcommand{\Tt}{\mathcal{T}}
\newcommand{\Ll}{\mathcal{L}_{\Nn}}
\newcommand{\Kk}{\mathcal{K}_{\Nn}}
\newcommand{\Mm}{\mathfrak{m}}
\newcommand{\SUPP}{\operatorname{supp}}
\newcommand{\HE}{\operatorname{HE}}
\newcommand{\KY}{\operatorname{K}}
\newcommand{\ED}{\mathfrak{e}}
\newcommand{\VE}{\mathfrak{v}}
\newcommand{\SO}{\operatorname{s}}
\newcommand{\RA}{\operatorname{r}}
\newcommand{\ID}{\operatorname{id}}
\newcommand{\HOM}{\operatorname{Hom}}
\newcommand{\HOL}{\operatorname{Hol}}
\newcommand{\Zz}{\mathbb{Z}}
\newcommand{\Cc}{\mathfrak{C}}
\newcommand{\CL}{\mathbb{C} \operatorname{ l}}
\newcommand{\Hhg}{\mathcal{G}}
\newcommand{\Gg}{\mathcal{G}}
\newcommand{\Ggamma}{\mathbf{\Gamma}}
\newcommand{\KKi}{\mathbf{\chi}}
\newcommand{\Ch}{\mathtt{Ch}_{^{_{\operatorname{JLO}}}}}
\newcommand{\ch}{\mathtt{ch}}
\newcommand{\hCh}{\mathtt{Ch}_{^{_{\operatorname{JLO}}}}}
\newcommand{\Ilim}{\varinjlim}
\newcommand{\Plim}{\varprojlim}
\newcommand{\LG}{\mathcal{LG}}
\newcommand{\ZO}{\diamond}
\newcommand{\KER}{\operatorname{ker}}
\newcommand{\DIFF}{\operatorname{diff}}
\newcommand{\ENDO}{\operatorname{End}}
\newcommand{\DOM}{\operatorname{Dom}}
\newcommand{\KCYCLE}{(\Bb,\Nn,\Dd)}
\newcommand{\FREDMOD}{(\rho,\Nn,F)}
\newcommand{\IND}{\mathtt{Ind}_\tau}
\newtheorem{definition}{Definition}[section]
\newtheorem{example}{Example}[section]
\newtheorem{theorem}{Theorem}[section]
\newtheorem{lemma}[theorem]{Lemma}
\newtheorem{corollary}[theorem]{Corollary}
\newtheorem{proposition}[theorem]{Proposition}
\newtheorem{remark}[definition]{Remark}

\setcounter{section}{-1}

\begin{abstract}
In attempts to combine non-commutative geometry
 and quantum gravity,
Aastrup-Grimstrup-Nest construct 
 a  semi-finite spectral triple,
 modeling
 the  space of $G$-connections for $G=U(1)$ or $SU(2)$.
AGN show that the interaction
 between the algebra of holonomy loops
 and the Dirac-type operator $\Dd$
reproduces the Poisson structure of General Relativity 
in Ashtekar's loop variables.
This article generalizes AGN's construction
to any connected compact Lie group $G$.
A construction of AGN's semi-finite spectral
triple in terms of an inductive limit of spectral triples 
is formulated. The refined construction permits
 the semi-finite spectral triple  to be
even when $G$ is even dimensional.
The Dirac-type operator $\Dd$ in
AGN's semi-finite spectral triple is a weighted sum of 
a basic Dirac operator on $G$. The weight
assignment is a diverging sequence that governs
the ``volume'' associated to each copy of $G$.
The JLO cocycle of AGN's triple is
examined in terms of the weight assignment. An explicit
condition on the weight assignment perturbations 
is given, so that the associated JLO class remains invariant.
Such a condition leads to
 a functoriality property of AGN's construction. 

\end{abstract}

\setcounter{tocdepth}{2}
\tableofcontents
\section{Introduction}
In non-commutative Geometry, a  space is 
represented by a $*$-algebra and the geometry  
on the space is given by
an unbounded self-adjoint operator $\Dd$ on a Hilbert space,
subject to certain axioms.
We call such a package a  spectral triple.
A typical example of a  spectral triple
is the Dirac triple
$\left(C^\infty(X),B(L^2(X,S)),\sDd \right)$ where $S$ is the
 spinor
 bundle of the spin manifold
$X$ and $\sDd$ is the Dirac operator of $X$
acting on $L^2(X,S)$. 
Connes' results \cite{reconstruction} state that 
geometric features of the manifold, such as metric, dimension,
differential forms, and integrations etc can be retrieved 
algebraically from the spectral triple. Therefore, 
a spectral triple gives a non-commutative notion of
manifolds when the
 given $*$-algebra is more general
than the space of functions on a manifold.
Recently, the notion of spectral triples is further generalized
 to model foliated manifolds or infinite dimensional 
manifolds that carry degeneracies \cite{type2index},
these generalized spectral triples are  called
semi-finite spectral triples. One
notable example is the non-commutative space of connections
by Aastrum-Grimstrup-Nest \cite{agn1}, which is the
main focus in this article.

In attempts to combine non-commutative geometry
 and quantum gravity,
Aastrup-Grimstrup-Nest construct 
 a  semi-finite spectral triple $(\Bb,\Nn,\Dd)$,
 modeling
 the  space of $G$-connections for the symmetry group 
 $G=U(1)$ or $SU(2)$. 
AGN show that the interaction
 between the algebra of holonomy loops $\Bb$
 and the Dirac-type operator $\Dd$
reproduces the Poisson structure of General Relativity 
in Ashtekar's loop variables \cite{agn2b,agn2,agn0,agn1},
they argue
that $(\Bb,\Nn,\Dd)$ incorporates quantum gravity in this model.

Unfortunately, building a (semi-finite) spectral triple
over the ordinary space of smooth connections 
like the Dirac triple
is impossible,
as there does not exist a Hilbert space structure
and Dirac operator on the infinite dimensional
affine space of connections $\Aa$.
One works instead with a
sequence of approximations of $\Aa$ by 
finite-dimensional manifolds.
In the Aastrup-Grimstrup-Nest approach,
 they compactify the space of connections over the manifold $M$
 by making  use of a finite graph together with its
refinements in   $M$ to
construct a separable kinematical Hilbert space and put a 
Dirac operator on it.
 Loosely speaking, the algebra $\Bb$
 is the pre-$C^*$ algebra of holonomies restricted to
the system of graphs in $M$, which mimics the holonomy algebra 
of Wilson loops;
  the densely defined operator $\Dd$
 is an infinite sum of a basic Dirac operator on
the symmetry group $G$ with appropriate weight 
assigned to each copy,
and $\Nn$ is a Type $\operatorname{II}_\infty$ von Neumann algebra
containing the CAR algebra.

For  technical reasons, Aastrup-Grimstrup-Nest
 limit their construction
to the symmetry group $G=U(1)$ or $SU(2)$.
In this article, we 
generalize AGN's construction to
any connected, compact Lie group $G$ by eliminating the
technical restriction.
Our construction also permits the semi-finite spectral triple to 
carry a $\mathbb{Z}_2$ grading when $G$ is even dimensional,
e.g. $G=SU(3)$,
the symmetry group that governs the strong force
in quantum field theory.
In a recent paper \cite{l}, the JLO character for semi-finite
spectral triples has been established. We will examine
the entire
cyclic cohomology class associated to the semi-finite 
spectral triple of AGN via the JLO character.
In particular, we give an explicit condition on allowable
perturbations of the given weight assignment so that
the associated JLO class remains invariant.
In a more recent paper \cite{agn3},
 Aastrup-Grimstup-Paschke-Nest 
specializes their construction to 
 lattice graphs,
which results in the most reasonable choice of 
weight assignment
that depends only on the dimension of the base manifold $M$.
If one re-runs AGN's construction on a sub-manifold
with dimension less than that of $M$,  
the resulting spectral triple will be
defined using the weight assignment 
corresponding to the sub-manifold,
which is 
 different from the
spectral triple obtained from pulling back the
construction on the  full manifold.
We prove a functoriality property of AGN's construction
by showing that the JLO cocycles 
of the pull back triple and the triple constructed
on a sub-manifold define the same entire cyclic cohomology class.

This paper is arranged as follows.
In Section~\ref{inductivelimitsofspectraltriples}, 
we will develop a limit for an inductive system of 
spectral triples. Section~\ref{AGN:construction}
gives an alternative construction of the
semi-finite spectral triple of  Aastrup-Grimstrup-Nest
using the formalism developed in
 Section~\ref{inductivelimitsofspectraltriples}.
In Section~\ref{ch:jlo}, we review the JLO theory
for semi-finite spectral triples developed in \cite{l}.
In Section~\ref{ch:jloandagn}, we examine the 
JLO class associated to AGN's semi-finite spectral triple
and the weak $\theta$-summability of the operator $\Dd$.


\section{Inductive Limit of Spectral Triples}
\label{inductivelimitsofspectraltriples}
\begin{definition}
\label{AGN:d:semifinitespectraltriple}
An  \textbf{\textit{odd} semi-finite spectral triple} $(\Bb,\Nn,\Dd)$
is  a (separable) semi-finite von Neumann algebra $\Nn\subset B(\Hh)$,
 $*$-sub-algebra $\Bb$ of $\Nn$, and a
densely defined unbounded self-adjoint operator $\Dd$
affiliated with $\Nn$ such that,
\begin{enumerate}
\item 
$[\Dd,b]$ extends to a
bounded operator for all $b\in \Bb$ ;
\item 
$(1+\Dd^2)^{-\frac{1}{2}}\in \Kk$ ,
where $\Kk$ is the ideal of $\tau$-compact operators in $\Nn$.
\end{enumerate}
If $(\Bb,\Nn,\Dd)$ is equipped with a 
$\mathbb{Z}_2$ grading $\chi\in\Nn$ such that
all $a$ is even for all $a\in\Bb$ and $\Dd$ is odd,
then we call $\KCYCLE$
 an \textbf{\textit{even} semi-finite spectral triple}.
The suffix \textit{semi-finite} is omitted 
when $\Nn=B(\Hh)$.

\end{definition}

\begin{proposition}[\cite{cprs2}]
Let $\Dd$ be an operator affiliated with $\Nn$,
and suppose that
 $T\in\Nn$ and that $[\Dd,T]$ is bounded. Then $[\Dd,T]\in\Nn$.
\end{proposition}

\begin{definition}
\label{triplemorphism}
A \textbf{morphism} from an even semi-finite spectral triple
$(\Bb,\Nn,\Dd)$ with grading $\chi\in \Nn$ to an even
 semi-finite spectral triple
 $(\Bb',\Nn',\Dd')$ with grading $\chi ' \in \Nn$ is
a triple  $(Q,P,\iota)$, where
\begin{enumerate}
\item 
 $\iota:\Hh\to \Hh'$ is a linear map
between the underlying Hilbert spaces of $\Nn$
and $\Nn'$
such that it preserves the inner products, and
$\iota(\DOM(\Dd))\subset \DOM(\Dd')$ so that
the following diagram commutes:
\begin{equation}
\label{intertwineiota}
\xymatrix{
\DOM(\Dd) \ar[rr]^{\Dd} 
\ar[dd]^{\iota}&& \Hh \ar[dd]^{\iota} \\
& \ar@(ul,dl)[]&\\
\DOM(\Dd') \ar[rr]^{\Dd'} && \Hh'
} \mbox{ ;}
\end{equation}
\item
 $P:\Nn\to\Nn'$ is a $*$-homomorphism so that
the following diagram commutes for all $a\in \Nn$:
\begin{equation}
\xymatrix{
\Hh \ar[rrr]^{a} 
\ar[dd]^{\iota}&&& \Hh \ar[dd]^{\iota} \\
&& \ar@(ul,dl)[]&\\
\Hh'\ar[rrr]^{P(a)} &&& \Hh'
} \mbox{ ;}
\end{equation}
\item
 $Q:\Bb\to\Bb'$ is a $*$-homomorphism so that
the following diagram commutes for all $b\in \Bb$:
\begin{equation}
\xymatrix{
\Hh \ar[rrr]^{b} 
\ar[dd]^{\iota}&&& \Hh \ar[dd]^{\iota} \\
&& \ar@(ul,dl)[]&\\
\Hh'\ar[rrr]^{Q(b)} &&& \Hh'
} \mbox{ ;}
\end{equation}
\item 
The following diagram commutes:
\begin{equation}
\xymatrix{
\Hh \ar[rrr]^{\chi} 
\ar[dd]^{\iota}&&& \Hh \ar[dd]^{\iota} \\
&& \ar@(ul,dl)[]&\\
\Hh'\ar[rrr]^{\chi'} &&& \Hh'
} \mbox{ .}
\end{equation}
\end{enumerate}
\end{definition}

A morphism between \emph{odd} semi-finite spectral triples is 
which  Definition~\ref{triplemorphism} with the last condition
dropped.

\begin{definition}
\label{directtriplesystem}
An \textbf{inductive system of semi-finite spectral triples}
 is an
$I$-family
of semi-finite spectral triples 
$\{(\Bb_j,\Nn_j,\Dd_j)\}_{j\in I}$
for a directed set $I$ and
together with a collection of morphisms
$\{(Q_{ij},P_{ij},\iota_{ij})\}_{i < j}$ so that the diagram 
\[
\xymatrix{
(\Bb_i,\Nn_i,\Dd_i)
\ar[rr]^{(Q_{ij},P_{ij},\iota_{ij})} 
\ar[rdd]_{(Q_{ik},P_{ik},\iota_{ik})}    &&
(\Bb_j,\Nn_j,\Dd_j) \ar[ldd]^{(Q_{jk},P_{jk},\iota_{jk})} \\
& \ar@(ul,ur)[]&&\\
& (\Bb_k,\Nn_k,\Dd_k)
}
\]
commutes for $i < j < k \in I$.
\end{definition}

Denote the limit of the Hilbert space systems
 $\{(\Hh_i,\iota_{ij})\}$ by $\Ilim \Hh_i$,
which is the Hilbert space closure of $\cup_{i\in I} \Hh_i$;
the limit of the $*$-algebra systems 
$\{(\Bb_i,Q_{ij})\}$ by $\Ilim \Bb_i$;
the limit of the von Neumann algebra 
systems $\{(\Nn_i,P_{ij})\}$
with underlying Hilbert space $\Ilim \Hh_i$
by $\Ilim \Nn_i$, which is the
weak operator closure of $\cup_{i\in I} \Nn_i$.
Since each $\Hh_j$ is a subspace of $\Ilim \Hh_i$,
each $\Dd_j$ on $\Hh_j$ extends to
an operator 
on $\Ilim \Hh_i$ 
by zero action on 
$\Hh_j ^\perp$.
Define the limit of the net of operators 
$\{ \Dd _j \}_{j\in I}$ 
acting on
$\Ilim \Hh_j$ to be
\begin{equation}
\label{limitofD}
(\Ilim \Dd_i) \eta:= \Ilim ( \Dd  _j  \eta )
\end{equation}
for $\eta$ in the \emph{appropriate domain} (to be clarified below) in $\Ilim \Hh_j$, where the right hand side is the 
limit of the net of vectors $\{\Dd_j \eta\}_{j\in I}\subset 
\Ilim \Hh_j$.
Denote by $\Ilim \chi_j$ 
the strong operator limit of the net of grading operators
$\{P_{j \infty}\chi_j\}_{j\in I}$ in $\Ilim \Nn_i$.
The following theorem justifies that $\Ilim \Dd_j$ and
$\Ilim \chi_j$ are well-defined operators on $\Ilim \Hh_j$.

\begin{theorem}
\label{propertiesoflimittriple}
Let $\{(\Bb_j,\Nn_j,\Dd_j) \}_{j\in I}$
be an inductive system of spectral triples.
Then
\begin{enumerate}
\item 
$\Ilim \Dd_j$
is an essentially self-adjoint operator on 
 $\Ilim
\Hh_j$. Denote its unique closure again by $\Ilim \Dd_j$.
\item 
$\Ilim\Dd_j$ is the strong resolvent limit of $\Dd_j$ and it
is affiliated with 
 $\Ilim
\Nn_j$.
\item
The commutator
\[
[\Ilim \Dd_j ,b]
\]
is bounded for all $b\in  \Ilim
\Bb_j$.
\item 
If each $(\Bb_j,\Nn_j,\Dd_j) $ is even equipped with a
grading operator $\chi_j$, then
\begin{itemize}
\item[a.]
$\Ilim \Dd_j$ anti-commutes with $\Ilim \chi_j$ and, \item[b.]
$b$ commutes with $\Ilim \chi_j$ for all $b\in \Ilim \Bb_j$.
\end{itemize}
\end{enumerate}
\end{theorem}
\begin{proof}
\mbox{ }
\begin{itemize}
\item[1,4a.]
Let $\eta \in \Ilim \DOM(\Dd_j)$,
then there exists $n\in I$ such that
$\eta \in  \DOM(\Dd_n) $.
Condition 1 of Definition~\ref{triplemorphism}
assures that the sequence $\{\Dd_j \eta\}$ stabilizes for
$j\geq n$.
We compute
\begin{eqnarray*}
(\Ilim \Dd_j) \eta& := &\Ilim ( \Dd_j \eta)\\
&=& \Dd_n \eta  \in \Ilim \Hh_j \mbox{ .}
\end{eqnarray*}
Therefore, $\Ilim (\Dd_j)$ is well-defined on 
$\Ilim \DOM(\Dd_j)$.
As $\DOM(\Dd_j)$ is dense in $\Hh_j$ for each $j$,
$\Ilim \DOM(\Dd_j)$  is dense in $\Ilim \Hh_j$.

Each $\Dd_j$ is self-adjoint on $\Hh_j$, thus
 the image $\operatorname{Im}(\Dd_j + \sqrt{-1}) =\Hh_j$.
Since $(\Ilim \Dd_j + \sqrt{-1})\eta = \Ilim 
((\Dd _j + \sqrt{-1} )\eta)$, the image $\operatorname{Im}
(\Ilim \Dd_j + \sqrt{-1})$ is the vector space limit 
$\Ilim \Hh_j$, which is dense in the Hilbert space 
limit $\Ilim \Hh_j$. As a result, $\Ilim \Dd_j$ is 
essentially self-adjoint.

Similarly, Condition 4 of Definition~\ref{triplemorphism}
justifies that $\Ilim \chi_j$ is well-defined on $\Ilim \Hh_j$.
On the other hand,
\begin{eqnarray*}
\lefteqn{(\Ilim \chi_j ) (\Ilim \Dd_i) \eta =
(\Ilim \chi_j) (\Dd_n \eta) = \chi_n \Dd_n \eta }\\&=&
- \Dd_n \chi_n \eta = (\Ilim \Dd_i)(\chi_n\eta)
=(\Ilim \Dd_i)(\Ilim \chi_j) \eta \mbox{ .}
\end{eqnarray*}
Hence $\Ilim \Dd_j$ is odd with respect to $\Ilim \chi_j$.
\item[2.]
By the ``point-wise'' construction
\eqref{limitofD}, $\Ilim \Dd_j$ is the strong graph limit 
of $\Dd_j$, which implies that $\Ilim \Dd_j$ is the strong resolvent
limit of $\Dd_j$ \cite{reedsimon}.
As $\Dd_j$ is affiliated with $\Nn_j$ for each $j$,
the sign and spectral projections of $\Dd _j$ 
are in $\Ilim \Nn_j$.
The fact that  $\Ilim \Nn_j$ is strong operator closed
implies that $\Ilim \Dd_j$ is affiliated with   $\Ilim \Nn_j$.
\item[3,4b.]
Let $b\in \Ilim \Bb_j$, then there exists $n$ so that
$b \in \Bb_n$.
We compute for  $\eta \in \Ilim \Hh_j$
\begin{eqnarray*}
[\Ilim \Dd_j, b] \eta &=& [\Dd_n,b]\eta 
\mbox{ ,}
\end{eqnarray*}
which is bounded since $(\Bb_n,\Nn_n,\Dd_n)$ is assumed
to be a spectral triple.

On the other hand,
\begin{eqnarray*}
[\Ilim \chi_j,b] \eta &=& [\chi_n,b]\eta =0 \mbox{ .}
\end{eqnarray*}
Hence $b$ is even with respect to $\Ilim \chi_j$.
\end{itemize}
\end{proof}

Notice that  $\Ilim \Dd_j$ being the strong 
resolvent limit of $\Dd_j$
 allows us to obtain functional calculus on $\Ilim \Dd_j$
as strong limits of functional calculi on $\Dd_j$.

\begin{theorem}[\cite{reedsimon}]
\label{thm:strongresolventthm}
Let $T_j$ and $T$ be self-adjoint operators such that
$T_j \rightarrow T$ in the strong resolvent sense. 
Then for any bounded continuous function $f$ on $\mathbb{R}$,
$f(T_j) \to f(T)$ in the strong operator limit.
\end{theorem}

\begin{definition}
Let $\{(\Bb_j,\Nn_j,\Dd_j)\}_{j\in I}$  be 
an inductive system of even semi-finite spectral triples
with grading operators $\{\chi_j\}_{j\in I}$,
define its limit to be
\[
\Ilim
(\Bb_j,\Nn_j,\Dd_j):=
\left(\Ilim
 \Bb_j,\Ilim
\Nn_j,\Ilim \Dd_j\right) \mbox{ .}
\]
It is equipped with the $\mathbb{Z}_2$ grading $\Ilim \chi_j$.
\end{definition}

The definition for the odd limit is obvious.

Unfortunately, the limit of a system of semi-finite
spectral triples needs not be a semi-finite spectral triple.
For instance, $I$ could be uncountable. Even if we assume $I$ to be countable
the limit may still not be a semi-finite spectral triple.
We will see both example and non-example from AGN's
construction of noncommutative connection space.
Nonetheless, the inductive limit of spectral triples satisfies
the following universal condition.

\begin{theorem}
\label{unversalityoflimittriple}
Let $\{(\Bb_j,\Nn_j,\Dd_j)\}_{j\in I}$ be 
an inductive system of even semi-finite spectral triples with
grading $\{\chi_j\}_{j\in I}$, and
suppose that $\{(\Bb',\Nn',\Dd')\}$ is an 
even semi-finite spectral triple with grading $\chi'$
such that there exist morphisms of spectral triples
$(Q'_j,P'_j,\iota'_j)$ such that the following diagram
commutes:
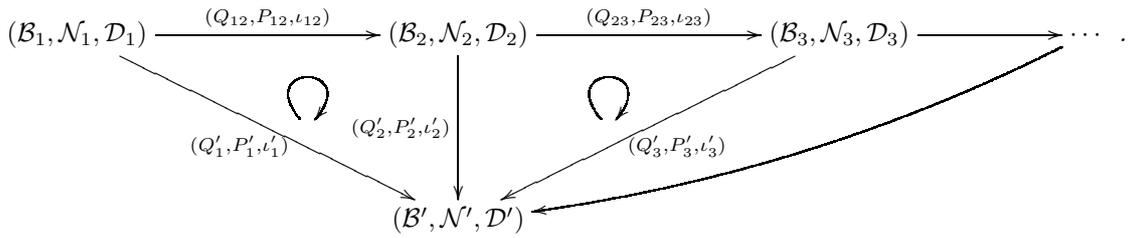
\[
\xymatrix{
(\Bb_1,\Nn_1,\Dd_1) \ar[rrr]^{(Q_{12},P_{12},\iota_{12})} 
\ar[ddrrr]_{(Q'_1,P'_1,\iota'_1)}
&&&
(\Bb_2,\Nn_2,\Dd_2) \ar[rrr]^{(Q_{23},P_{23},\iota_{23})}
\ar[dd]_{(Q'_2,P'_2,\iota'_2)} &&& 
(\Bb_3,\Nn_3,\Dd_3) \ar[rr] 
\ar[ddlll]^{(Q'_3,P'_3,\iota'_3)}
&& \cdots \ar@/^1pc/[ddlllll]\\
& &\ar@(ul,ur) & & \ar@(ul,ur) && \\
&&&(\Bb',\Nn',\Dd')&&&
}\mbox{ .}
\]
Then there exists a unique morphism $(Q,P,\iota)$ 
completing the following diagram:
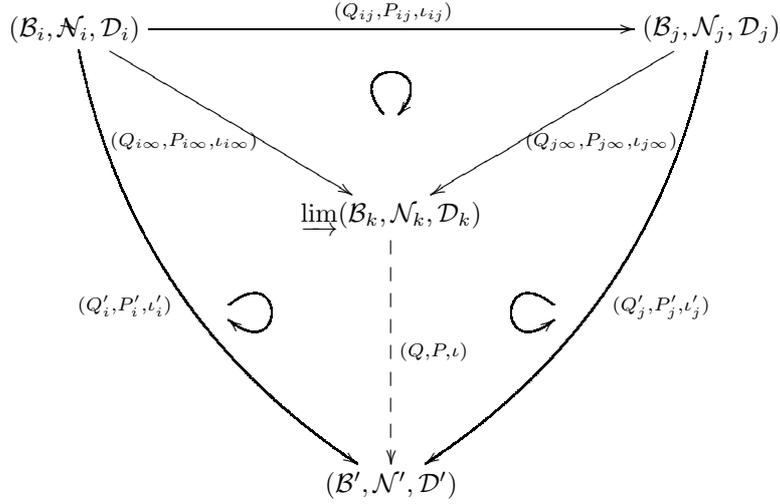
\[
\xymatrix{
(\Bb_i,\Nn_i,\Dd_i) \ar[rrrr]^{(Q_{ij},P_{ij},\iota_{ij})} 
\ar[ddrr]_{(Q_{i\infty},P_{i\infty},\iota_{i\infty})}
\ar\ar@/_2pc/[dddddrr]_{(Q'_{i},P'_{i},\iota'_{i})}
&&&&
(\Bb_j,\Nn_j,\Dd_j)
\ar[ddll]^{(Q_{j\infty},P_{j\infty},\iota_{j\infty})}
\ar@/^2pc/[dddddll]^{(Q'_{j},P'_{j},\iota'_{j})}
\\
&& \ar@(ul,ur) &&\\
&& \Ilim (\Bb_k,\Nn_k,\Dd_k) \ar@{-->}[ddd]^{(Q,P,\iota)}
 && \\
& \ar@(ur,dr)
&&
\ar@(ul,dl)
&\\
&&&&\\
&& (\Bb',\Nn',\Dd')  &&
}
\]
for all $i,j\in I$.
\end{theorem}
\begin{proof}
The existence and uniqueness of 
$Q$, $P$, and $\iota$ come from the universalities of
$\Ilim \Bb_k$, $\Ilim \Nn_k$, and $\Ilim \Hh_k$.
$\iota(\Ilim \DOM(\Dd_k))\subset \DOM(\Dd')$ and the 
existence of the commutative diagrams
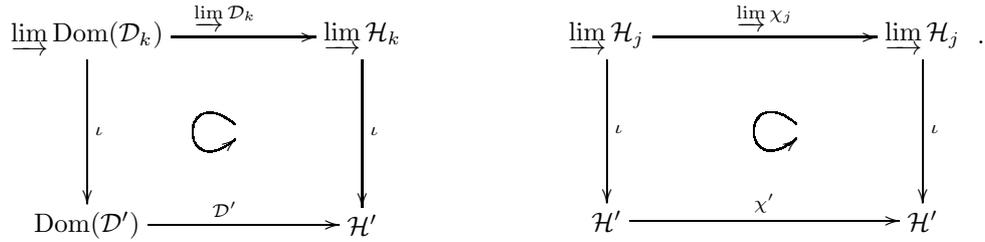
\[
\xymatrix{
\Ilim \DOM(\Dd_k) \ar[rr]^{\Ilim \Dd_k}
\ar[dd]^{\iota}
 &&  \Ilim \Hh_k \ar[dd]^{\iota} \\
&\ar@(ul,dl) &\\
 \DOM(\Dd') \ar[rr]^{\Dd'} && \Hh' \\
}
\hspace{2cm}
\xymatrix{
\Ilim \Hh_j \ar[rrr]^{\Ilim \chi_j} 
\ar[dd]^{\iota}&&& \Ilim \Hh_j \ar[dd]^{\iota} \\
&& \ar@(ul,dl)[]&\\
\Hh'\ar[rrr]^{\chi'} &&& \Hh'
} \mbox{ .}
\]
follow by constructions.
\end{proof}

\section
[The Noncommutative Space of Connections of
AGN]
{The Noncommutative Space of Connections of
Aastrup-Grimstrup-Nest}
\label{AGN:construction}

In this section, we give an alternative construction 
for AGN's semi-finite spectral triple 
 that models the space of $G$-connections.
We make use of the inductive limit formalism 
developed in Section~\ref{inductivelimitsofspectraltriples}
to formulate  AGN's triple as a limit of a sequence
 well-behaved spectral triples.
In Section~\ref{AGN:compactification},
 we will construct a spectral triple on a graph.
Then in Section~\ref{compactification}, we follow AGN's idea to
 compactify the space of connections
using a graph and its refinements in a manifold.
Section~\ref{AGN:dirac} constructs 
a corresponding system of spectral triples
from the system of  a graph and its refinements.
Section~\ref{AGN:semifinite} will alter the spectral triple system 
constructed in Section~\ref{AGN:dirac}
appropriately so that the limit of the system is a 
semi-finite spectral triple,
 and discuss 
the grading operator on the limit spectral triple.

\subsection{Graphs and Spectral Triples on Graphs}
\label{AGN:compactification}

\subsubsection{Space of Connections on Graphs}

\begin{definition}
\label{AGN:d:graph}\mbox{ }

\begin{itemize}
\item 
A \textbf{directed graph} $\Gamma$ is a set $V_\Gamma$ (vertices)
and a set $E_\Gamma$ (edges) with two maps
 $\SO,\RA: E_\Gamma \to V_\Gamma$ (source and range).
\item
A \textbf{morphism of graphs} $\Gamma \to \Gamma'$ consists
of maps
$E_\Gamma \to E_{\Gamma'}$ and $V_\Gamma \to V_{\Gamma'}$ so 
that they intertwine the source and range maps.
\item
We call a directed graph $\Gamma$ \textbf{finite}
if the sets $E_\Gamma$ and $V_\Gamma$ have finite cardinalities.
\end{itemize}

\end{definition}

We view the vertices as a collection of points,
 and the edges as arrows
from $\SO(e)$ to $\RA(e)$.

\begin{example}\mbox{ }
\begin{itemize}
\item 
Any groupoid $\Gg  \rightrightarrows ^{\hspace{-0.3cm}^{\RA}}
_{\hspace{-0.3cm}_{\SO}}X$ is a directed graph.
\item 
Any subset of a groupoid $\Gg$ is a directed graph.
\item
Any set $S$ can be viewed as a graph by taking
$E_\Gamma:=S$ and $V_\Gamma=\{\operatorname{pt}\}$.
\end{itemize}
\end{example}

\begin{theorem}
\label{universalgraph}
Given a directed graph $\Gamma$, there is a unique
groupoid $\Gg(\Gamma)  \rightrightarrows ^{\hspace{-0.3cm}^{\RA}}
_{\hspace{-0.3cm}_{\SO}} V_\Gamma$ with 
a graph morphism $\Gamma \to \Gg(\Gamma)$,
so that given 
any groupoid $\Gg  \rightrightarrows ^{\hspace{-0.3cm}^{\RA '}}
_{\hspace{-0.3cm}_{\SO '}} X$ with a graph morphism
$\Gamma \to \Gg$, there exists a unique groupoid morphism
$\Gg(\Gamma) \dashrightarrow \Gg$ such that the 
following 
 diagram commutes:
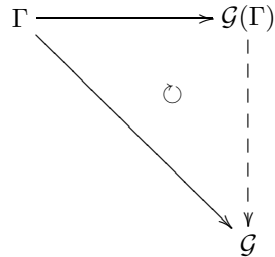
\[
    {
\xygraph{
!{<0cm,0cm>;<1cm,0cm>:<0cm,1cm>::}
!{(1,4) }*+{\Gamma}="14"
!{ (4,4) }*+{\Gg(\Gamma)}="44"
!{ (4,1) }*+{\Gg}="41"
!{ (3,3) }*+{\circlearrowright}="c"
"14":"44"
"44":@{-->}"41"
"14":"41"
}
} \mbox{ .}
\]
$\Gg(\Gamma)$ is called the \textbf{free groupoid generated by $\Gamma$}.
\end{theorem}

\begin{remark}
If $V_\Gamma=\{\operatorname{pt}\}$, then  $\Gamma$ is just a set given by
$E_\Gamma$.
In that case, $\Gg(\Gamma)$ is the \textbf{free group}
generated by the set $E_\Gamma$.
\end{remark}

\begin{definition}
\label{freegeneratingset}
Let $\Gg(\Gamma)$ be the free groupoid generated by $\Gamma$.
The subset $\Ff_{\Gamma}\subset \Gg(\Gamma)$ is called the 
\textbf{free generating set of $\Gg(\Gamma)$} if given 
any groupoid $\Gg  \rightrightarrows ^{\hspace{-0.3cm}^{\RA '}}
_{\hspace{-0.3cm}_{\SO '}} X$  and a set map $\Ff_\Gamma \to \Gg$, 
there exists a unique groupoid morphism $\Gg(\Gamma)\to\Gg$ such 
that the following diagram commutes:
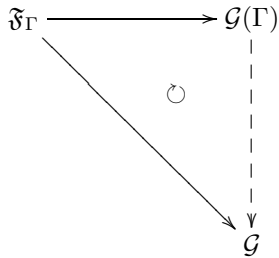
\[
    {
\xygraph{
!{<0cm,0cm>;<1cm,0cm>:<0cm,1cm>::}
!{(1,4) }*+{\Ff_\Gamma}="14"
!{ (4,4) }*+{\Gg(\Gamma)}="44"
!{ (4,1) }*+{\Gg}="41"
!{ (3,3) }*+{\circlearrowright}="c"
"14":"44"
"44":@{-->}"41"
"14":"41"
}
} \mbox{ .}
\]
\end{definition}

\begin{example}
\mbox{ }
\begin{itemize}
\item 
Let $\Gamma$ be a graph, then $E_{\Gamma}\subset \Gg(\Gamma)$
is a free generating set. 
\item 
Let $\Gamma$ be the graph
\begin{eqnarray*}
\xymatrix{
 \Gamma:& \bullet \ar@{-}[rr]_{e^{-}}|-{\object@{>}} &&\bullet
 \ar@{-}[rr]_{e^{+}}|-{\object@{>}} &&\bullet  \mbox{ ,}
 }  \end{eqnarray*}
then the set composed of the paths $(e^{-},e^{+})$ and $(e^{+})$
forms a free generating set for $\Gg(\Gamma)$.
\end{itemize}\end{example}

\begin{remark}
\label{equalcardinality}
Let $\Gamma$ be a finite graph, then 
 cardinality of $\Ff_{\Gamma}$ equals that of $E_{\Gamma}$.
Since every free generating set necessarily has the same cardinality.
\end{remark}

\begin{proposition}
\label{groupoidgraphhomo}
For any groupoid $\Gg$, there is a 
 bijection between
the set of groupoid homomorphisms 
$\Gg(\Gamma)\to
\Gg$
and 
the set of maps
$\Ff_\Gamma \to\Gg$.
\end{proposition}
\begin{proof}
It follows from Definition~\ref{freegeneratingset}
that
the space of set maps  $\MAP(\Ff_\Gamma,\Gg)$
injects into the space of set  homomorphisms $\HOM(\Gg(\Gamma),\Gg)$.
As sets, $\Ff_\Gamma\subset \Gg(\Gamma)$. Hence
every groupoid homomorphism in $\HOM(\Gg(\Gamma),\Gg)$
restricts to a set map in $\MAP(\Ff_\Gamma,\Gg)$.
\end{proof}

\begin{corollary}
For any groupoid $\Gg$, there is a 
 bijection between
the set of groupoid homomorphisms 
$\Gg(\Gamma)\to
\Gg$
and 
the set of graph homomorphisms 
$\Gamma\to\Gg$.
\end{corollary}
\begin{proof}
Choose the generating set $\Ff_\Gamma$ to be the set of edges $E_\Gamma$.
\end{proof}


\begin{definition}
\label{AGN:d:coarsegrained}
Given a group $G$, define the
\textbf{space of $G$-connections on
 $\Gamma$} to be
\[
\Aa_\Gamma:=\HOM(\Gg(\Gamma),G)\mbox{ ,}
\]
where $\HOM$ is understood to be the space of 
groupoid homomorphisms.
\end{definition}
\begin{corollary}
\label{AGN:lem:generatoridentification}
Let $\Gamma$ be a finite graph, then
there exists a  bijection 
\begin{equation}
\label{AGN:eqn:generatoridentification}
\Aa_\Gamma \stackrel{\sim}{\longrightarrow}
 G^{\lvert E_\Gamma \rvert}
\end{equation}
\end{corollary}
\begin{proof}
Proposition~\ref{groupoidgraphhomo} specializes to
$\Aa_\Gamma \cong \MAP(\Ff_\Gamma,G)$.
By Remark~\ref{equalcardinality},
$\MAP(\Ff_\Gamma,G)=G^{\lvert E_\Gamma \rvert}$.
The result is obtained.
\end{proof} 

$G$ is thought of as a symmetry group, so
we will assume $G$ to be a compact Lie group,
and it comes equipped with 
the  normalized Haar measure.
We equip $\Aa_\Gamma$ the manifold structure and measure
coming from  $G^{\lvert E_\Gamma \rvert}$ under the 
identification 
\eqref{AGN:eqn:generatoridentification}.
Hence $\Aa_\Gamma$ is a 
compact manifold with a smooth measure.
Note that the manifold structure on $\Aa_\Gamma$ depends on 
the choice of free generating set $\Ff_\Gamma$. However,
we believe that the measure on $\Aa_\Gamma$ is intrinsic, i.e.
it does not depend on the choice of generating set.
Most of the time we think of the generating is just $E_\Gamma$.
It will be made clear in later sections that when and why we
generalize to other free generating sets.

The gauge group of the connection space $\Aa_\Gamma$ is
$\GAU_\Gamma:=\MAP(V_\Gamma,G)$. 
The  action of $\GAU_\Gamma$ on $\Aa_\Gamma$ is given by
\[
g (\nabla ) (\gamma) :=  
g\left( \SO(\gamma) \right) \cdot
  \nabla(\gamma) \cdot g \left(\RA(\gamma)\right)^{-1}
\]
for $g \in \GAU_\Gamma, \gamma\in \Gg(\Gamma)$, 
and $\nabla\in\Aa_\Gamma$. $\GAU_\Gamma$ preserves the measure
on $\Aa_\Gamma$.


\begin{definition}
\label{AGN:d:isotropy}
Fix a vertex $\nu$ of a graph $\Gamma$.
Define the \textbf{isotropy group
 of $\Gamma$ at
$\nu$}
to be the group
\[
\Hhg_\nu(\Gamma):=\left\{
\gamma \in \Hhg(\Gamma) : 
\SO(\gamma)=\nu=\RA(\gamma) \right\} \mbox{ .}
\]
\end{definition}

Suppose that $\gamma\in \Gg(\Gamma)$ such that
 $\SO(\gamma)=\nu$ and $\RA(\gamma)=\mu$ for some vertices
$\nu,\mu \in V_\Gamma$. Then
it is easy to see that $\Gg_\nu(\Gamma)$ and $\Gg_\mu(\Gamma)$
are isomorphic as groups with the isomorphism given by conjugation
by $\gamma$. 

\begin{definition}
A graph $\Gamma$ is said to be \textbf{connected}
if the groupoid $\Hhg(\Gamma)$ is transitive.
That is, 
for every pair of vertices $\nu,\mu\in V_\Gamma$, there exists
an element $\gamma\in\Gg(\Gamma)$
such that $\SO(\gamma)=\nu$ and $\RA(\gamma)=\mu$.
\end{definition}

We will assume that all the graphs we are dealing with
are connected.

\subsubsection{The Algebra of Holonomies}
Each $\gamma\in \Hhg_\nu(\Gamma)$ defines a \emph{smooth}
 $G$-valued
function on $\Aa_\Gamma$ given by 
\begin{equation}
\label{embedisotropy}
h_\gamma (\nabla) := \nabla(\gamma) \mbox{ , for } \nabla\in \Aa_\Gamma=\HOM(\Hhg(\Gamma),G) \mbox{ .}
\end{equation}
Thus, 
\[
h: \Hhg_\nu(\Gamma) \to C^\infty(\Aa_\Gamma,G)\mbox{ .}
\]
In fact, $h$ is a group homomorphism, as
\[
h_{\gamma ' \circ \gamma}( \nabla) =
\nabla (\gamma ' \circ \gamma) =
\nabla(\gamma')\cdot \nabla(\gamma)=
h_{\gamma'}(\nabla)\cdot h_ \gamma(\nabla) \mbox{ ,}
\]
where the product on $ C^\infty(\Aa_\Gamma,G)$ is given pointwise.
$h$ is the inverse of the \emph{loop transform} in Loop Quantum Gravity \cite{rovelli}.
\begin{definition}
\label{holonomyalgebraofgamma}
Define $\Bb_\Gamma$, the \textbf{algebra of $\Gamma$-holonomies}, to 
 be the 
group algebra generated by the subgroup
$h(\Hhg_\nu(\Gamma))\subset C^\infty(\Aa_\Gamma,G)$.
\end{definition}
 An
element of $\Bb_\Gamma$
is a finite sum of elements in $h\left(\Hhg_\nu(\Gamma)\right)$
 with complex coefficients:
\[
a=\sum_i a_i h_{\gamma_i} \mbox{ ,}
\]
where $a_i\in \mathbb{C} $
 and $h_{\gamma_i}\in h(\Hhg_\nu(\Gamma))$.

\begin{remark}
The algebra $\Bb_\Gamma$ does not
depend on the manifold structure on $\Aa_\Gamma$, as
the inverse loop transform \eqref{embedisotropy} does not.
Therefore, $\Bb_\Gamma$ is independent of how 
$\Aa_\Gamma$ is identified with $G^{|E_\Gamma|}$ under different
free generating sets $\Ff_\Gamma$ for $\Gg(\Gamma)$.
\end{remark}

Suppose that $G$ comes with a faithful unitary representation
as matrices in $\MATRIX$.
Then $\Bb_\Gamma$ is a $*$-subalgebra of 
$C^\infty(\Aa_\Gamma,\MATRIX)$,
with   the involution  given by
\[a^*:=\sum_i \overline{a_i} h_{\gamma_i}^{-1} \mbox{ .}\]
 $\Bb_\Gamma$ is a pre-$C^*$ algebra with norm inherited
from  $C^\infty(\Aa_\Gamma,\MATRIX)$.
As a subalgebra of $C^\infty(\Aa_\Gamma,\MATRIX)$,
$\Bb_\Gamma$ represents on the 
Hilbert space $L^2(\Aa_\Gamma,
E
)$ via point-wise
multiplication:
\begin{eqnarray}
\label{brep}
(a \cdot \eta ) \nabla := a (\nabla) \cdot \eta(\nabla) 
\end{eqnarray}
for $a \in \Bb_\Gamma\subset C^\infty(\Aa_\Gamma,\MATRIX)$,
 $\eta\in L^2(\Aa_\Gamma,
E
)$, $\nabla\in \Aa_\Gamma$, where $E$ is any
(finite dimensional) 
$\MATRIX$-module.

\subsubsection{Quantum Weil Algebra} 
\label{DH}
Denote by $\g$ the Lie algebra of $G$, let $\Uu(\g)$
be the universal enveloping algebra of $\g$. Fix 
an invariant metric $\langle \cdot, \cdot \rangle$ on $\g$
 and let $\CL(\g)$ be the Clifford
algebra generated by $\g$ with respect to the
relation $a b + b a = 2 \langle a , b \rangle $ for $a,b\in \g$.
\begin{definition}
\label{AGN:d:weilalg}
Define the \textbf{quantum Weil algebra}  
$\Ww(\g)$  of a Lie algebra $\g$ to be
\[
\Ww(\g):=\Uu( \g ) \otimes \CL(\g)
\mbox{ .}
\]
\end{definition}

Notice that $\CL(\g)$ is a finite dimensional $C^*$-algebra.
Let 
\[
\CL(\g) \rightarrow B(S)
\]
be a cyclic representation of $\CL(\g)$ 
on a Hilbert space $S$ with normalized
cyclic vector $\mathbb{I}$,
i.e.,
$\CL(\g)\mathbb{I} = S$.
Examples of such $S$ are the exterior algebra
$\wedge ^\bullet (\g)$ with 
 $\mathbb{I}$ given by the $1$ in $\wedge ^0 (\g)$,
and certain quotients of
$\wedge ^\bullet (\g)$, such as the spin representations.
We get an action of $\Ww(\g)$ on $L^2(G)\otimes S$,
where $\Uu(\g)$ acts on $L^2(G)$ as left-invariant
differential operators.

Let $\left\{e_i\right\}_{i=1}^q$ be an orthonormal basis of $\g$ with respect to the chosen
invariant metric.
Define the Dirac operator on 
$L^2(G)\otimes S$
to be
the element
 $D\in \Ww(\g)$ given by
\begin{equation}
\label{AGN:eqn:basicdirac}
D:= \frac{1}{\sqrt{-1}} \sum _{i=1}^q  e_i \otimes e_i
 \mbox{ .}
\end{equation}

$D$ is essentially self-adjoint.
Denote its unique self-adjoint extension again by $D$.

Aastrup-Grimstrup-Nest proved the following result for 
$G=SU(2)$ \cite{agn1}.
\begin{theorem}
\label{AGN:thm:eckhard}
Let $G$ be a compact  Lie group $G$.
The operator $D$ on $L^2(G)\otimes S$ has kernel
\[\ker(D) = \mathbb{C} \otimes S \mbox{ ,}\]
where $\mathbb{C}\subset L^2(G)$ is embedded as constant functions.
\end{theorem}
The proof will make use of the following results by Kostant \cite{kostant}.

\begin{lemma}[\cite{kostant}]
\label{l:fromgtoclg}
The map $\pi:\g \to \CL (\g)$ given in the orthonormal basis 
$e_i\in \g$ by
\begin{equation}
\label{eq:fromgtoclg}
 \pi(e_a)= 
-\frac{1}{4}\sum_{i,j}\langle [e_i,e_j],e_a \rangle e_i e_j
\end{equation}
is a Lie algebra homomorphism.
\end{lemma}

Fix a Cartan subalgebra  and a system of positive roots of $\g$.
\begin{lemma}[\cite{kostant}]
\label{l:cliffrhorep}
Let
$S$ be any $\CL(\g)$-module. Then the $\g$-representation 
on $S$ defined by
composition with $\pi$ is a direct sum of $\rho$-representations,
 where $\rho$ is the half sum of all positive roots. 
\end{lemma}

Denote by $\CAS$ the element $-\sum_{i} e_i e_i \in \Uu(\g)$.
$\CAS$ is called the Casimir \cite{knapp}.

\begin{lemma}[\cite{knapp}]
\label{lem:casscalar}
Let $V$ be an irreducible representation of $\g$
 with 
highest weight $\lambda$.
Then $\CAS$ acts as a scalar on $V$, and the scalar is given by
\[\lvert \lambda +\rho \rvert ^2 - \lvert \rho \rvert^2 ,\]
where $\rho$ is the half sum of all positive roots.
\end{lemma}

\begin{proof}[Proof of Theorem~\ref{AGN:thm:eckhard}]
As $\ker(D)=\ker( D  ^2)$, we compute
\[
 D  ^2=- \sum_{i,j} e_i e_j \otimes e_i e_j
= - \sum_k e_k e_k \otimes 1 - \frac{1}{2}\sum_{i,j} 
[e_i, e_j ]\otimes e_i e_j \mbox{ .}
\]
%
Let
\[
 \pi : \Uu(\g) \to \CL(\g)
\]
be the lift of the Lie algebra homomorphism 
\eqref{eq:fromgtoclg} defined in Lemma~\ref{l:fromgtoclg}.
Then $D^2$ becomes
\[
D^2=\CAS \otimes 1  + 2 \sum_{k=1}^q e_k \otimes   \pi  (e_k)\mbox{ .}
\]
On the other hand, 
\begin{eqnarray}
\label{deltacas}
\Delta(\CAS)=\CAS \otimes 1 + 1 \otimes \CAS - 2
 \sum_{k=1}^q e_k \otimes e_k 
\mbox{ ,}
\end{eqnarray}
where \[\Delta:\Uu(\g)\to \Uu(\g \oplus \g)=
\Uu(\g)\otimes \Uu(\g)\]
is the co-multiplication given by the diagonal embedding,
which is a Lie algebra homomorphism.
Using 
\eqref{deltacas}, we can write $D^2$ in terms of the Casimir:
\[
D^2=   2\CAS \otimes 1 +    1 \otimes  \pi (\CAS)  -
(1\otimes  \pi)\Delta (\CAS)  \mbox{ .}
\]
Let
\[T:= 2\CAS \otimes 1 +    1 \otimes  \CAS  -
\Delta \CAS \in \Uu(g)\otimes \Uu(\g), \] so 
that
 the action
of $D^2$ on $L^2(G)\otimes S$
 coincides with the action of $(1\otimes\pi)(T)\in
\Uu(\g)\otimes \CL(\g)$.
As a
$\g\oplus \g$-representation,
$ L^2(G)\otimes S$ is a direct sum of components
$V_\lambda\otimes V_\rho$, for dominant weights $\lambda$
 by  Lemma~\ref{l:cliffrhorep}.
 We will now
determine the smallest eigenvalue of
$T$ 
on $V_\lambda\otimes V_\rho$ .
The highest weights of the irreducible components
for the diagonal action are all less than or equal to $ \lambda+\rho$.
 Thus by Lemma~\ref{lem:casscalar}, the action of $T$ on $V_\lambda \otimes V_\rho$
is bounded below by
\begin{eqnarray*} 
\lefteqn{
2\left(\lvert \lambda + \rho \rvert ^2 - \lvert \rho \rvert^2 \right)
+ \left( \lvert \rho +\rho \rvert ^2 - \lvert \rho \rvert^2 \right)
- \left( \lvert \lambda + \rho +\rho \rvert ^2  -
 \lvert \rho \rvert ^2\right)} \\
& =&
2 \lvert \lambda + \rho \rvert ^2  +2  \lvert \rho \rvert^2 
-  \lvert \lambda + 2\rho \rvert   ^2 \\
& =& 
2\lvert \lambda \rvert^2  + 4 \langle \lambda,\rho \rangle 
 + 4\lvert \rho \rvert ^2 - 
 \lvert \lambda \rvert ^2 -
4 \langle \lambda, \rho \rangle
- 4 \lvert \rho \rvert^2
 = \lvert 
\lambda \rvert ^2
 \mbox{ .} 
\end{eqnarray*}
The bound is strictly positive unless $\lambda=0$.
Furthermore, the space $V_{0}$ has multiplicity $1$ and
  embeds in $L^2(G)$  as the
constant functions.
Hence, we conclude that the $0$-eigenspace of $D$ is precisely
the space of constant functions $\mathbb{C}\otimes S$.
\end{proof}

\begin{remark}
\label{remark:smallesteigenvalue}
In the case that $G$ is semi-simple and simply connected,
the smallest non-zero eigenvalue of $D^2$ is given by 
$\lvert \rho \rvert ^2$. When the invariant metric 
$\langle \cdot,\cdot \rangle$ is chosen to be the Killing form, 
a theorem due to  Freudenthal and de Vries \cite{mein} states that
\[\lvert \rho \rvert ^2 = \frac{1}{24} 
\TR (\operatorname{Ad} \CAS)=\frac{\dim (\g)}{24} \mbox{ .}\] 
\end{remark}

\subsubsection{The Dirac Operator and the Hilbert Space}
Let $\Gamma$ be a finite graph.
To each
element $\gamma\in \Gg(\Gamma)$,
we associate to it a Hilbert space $H_\gamma:=L^2(G)\otimes S$
and an operator $D_\gamma:=D$ on $H_\gamma$. 
Intuitively speaking, we are associating each path $\gamma$
a copy of $G$, which is thought of as coming from 
holonomies of $G$-connections along $\gamma$.
Define the Hilbert space $\Hh_{\Ff_\Gamma}$ to  be
\[
\Hh_{\Ff_\Gamma}:=\bigotimes_{\gamma\in \Ff_\Gamma} H_\gamma  
\cong
  L^2(G^{|E_\Gamma|})\otimes 
S^{\otimes \lvert E_\Gamma \rvert }\] and the
Dirac operator  $\Dd_{\Ff_\Gamma}$ on $\Hh_{\Ff_\Gamma}$ to be
\begin{equation}
\label{DiraconGamma}
\Ww ( \g^{ \lvert E_\Gamma \rvert})\ni 
\Dd_{\Ff_\Gamma}:=\sum_{\gamma\in \Ff_\Gamma}D_\gamma
\end{equation}
with 
$D_\gamma$ being the
obvious extension to $\Hh_{\Ff_\Gamma}$,
where $\Ff_\Gamma$ is the set of free generators
of the groupoid $\Gg(\Gamma)$, and it has cardinality 
equals to the number of edges $|E_\Gamma|$.
When the set of free generators $\Ff_\Gamma$ is $E_\Gamma$, we
will denote the corresponding Hilbert space and Dirac operator by
$\Hh_\Gamma$ and $\Dd_\Gamma$ respectively.
At this stage, it may seem unclear if one ever needs
the case of $\Ff_\Gamma$ not being $E_\Gamma$.
It will be made clear in Section~\ref{AGN:dirac}
that this generalized definition is essential in constructing
a limit Dirac operator.


We extend the action of the Quantum Weil algebra
 $\Ww(\g ^{|E_\Gamma |})$ to
 $E\otimes \Hh_{\Ff_\Gamma}$ by letting it act as identity on
the $\MATRIX$-module $E$,
 and extend the
action of the   algebra of $\Gamma$-holonomies $\Bb_\Gamma$ on
$E\otimes \Hh_{\Ff_\Gamma}$
to act as identity on $S^{|E_\Gamma|}$.


\begin{proposition}
\label{graphtriple}
For any free generating set $\Ff_\Gamma$,
the triple
$(\Bb_\Gamma,B(E\otimes \Hh_{\Ff_\Gamma}),\Dd_{\Ff_\Gamma})$ is a spectral triple.
\end{proposition}
\begin{proof}
$\Dd_{\Ff_\Gamma}$ is a formally self-adjoint elliptic differential operator on 
$G^{|E_\Gamma|}$. Hence it is essentially self-adjoint and has
 compact resolvent.
The fact that 
$[\Dd_{\Ff_\Gamma}, b]$ extends to a bounded operator on
 $\Hh_{\Ff_\Gamma}$ 
for all $b\in \Bb_\Gamma$ comes from the fact that
$\Bb_\Gamma$ is a sub-algebra of $C^\infty(\Aa_\Gamma,\MATRIX)$.
\end{proof}

\subsection{A Compactification of the Space of Connections}
\label{compactification}

\subsubsection{Systems of Graphs}
\begin{definition}
\label{AGN:defn:subgraph}
A \textbf{refinement} $\Gamma'$ of $\Gamma$,
denoted $\Gamma < \Gamma'$,
is a graph homomorphism $\Gamma \to \Gg(\Gamma')$
such that
 the image of every $e$ of $E_\Gamma$ in $\Gg(\Gamma')$ is 
a product of elements in $E_{\Gamma'}$ and 
the induced groupoid homomorphism
$\Gg(\Gamma) \dashrightarrow \Gg(\Gamma')$ is injective.
\end{definition}

\begin{definition}
\label{AGN:d:directedsystemofgraphs}
Let $\Ss=\{\Gamma_i\}_{i\in I} $ be a family of graphs
indexed by a directed set $I$, then we call $S$ a 
\textbf{directed system of graphs}
if for $\Gamma_i,\Gamma_j \in S$, one has
$\Gamma_i,\Gamma_j < \Gamma_k$ 
 for
$k\in I$ with $i,j < k$.
\end{definition}

Notice that a directed system of graphs is itself
a directed set. The original compactification 
of the space of smooth $G$-connections by Ashtekar-Lewandowski
uses the set of all embedded finite graphs as the directed set
\cite{rovelli}.

By the definition of graph refinements
$\Gamma_1 < \Gamma_2$ (Definition~\ref{AGN:defn:subgraph}),
a directed system of graphs $\{\Gamma_i\}_{i\in I} $
 gives rise to an inductive system of  groupoids
$\{\Gg(\Gamma_i)\}_{i\in I} $ with connecting morphisms being
 the groupoid inclusions.
This system of groupoids  
gives rise to a projective system of connection spaces
$\{ \Aa_{\Gamma_i}\}_{i\in I} $.
When the graphs are finite and the free generating set 
$\Ff_{\Gamma_i}$ of $\Gg(\Gamma_i)$ is chosen to be $E_{\Gamma_i}$
for each $i\in I$,
the morphisms 
\[
\Aa_{\Gamma_i} \twoheadleftarrow \Aa_{\Gamma_{j}}
\mbox{  for } i<j
\]
consist of \emph{projections}, and  \emph{multiplications} 
of Lie groups,
under the identification 
\eqref{AGN:eqn:generatoridentification}.
Hence, they are
\emph{smooth surjective submersions} and \emph{preserve the Haar measures}.
Therefore, 
$\{ \Aa_{\Gamma_i}\}_{i\in I} $ is a
projective system of topological measure spaces.

\begin{definition}
\label{AGN:d:generalconnectionspace}
Let $\Ss=\{\Gamma_i\}_{i\in I} $ a directed system
of finite graphs.
Define the \textbf{space of  generalized $G$-connections}
$\overline{\Aa}^\Ss$
to be the projective limit 
\[
\overline{\Aa}^\Ss :=\Plim \{ \Aa_{\Gamma_i}\}_{i\in I} 
\mbox{ .}
\]
\end{definition}

\begin{proposition}[\cite{limitmeasure}]
\label{AGN:p:generalconnectionspace}
The space of generalized $G$-connections
$\overline{\Aa}^\Ss$ 
is a connected compact Hausdorff measure space.
\end{proposition}
The limit measure on $\overline{\Aa}^\Ss$ 
is called the Ashtekar-Lewandowski measure \cite{rovelli}.

\begin{proposition}[\cite{groupoidconnection}]
\label{limitofhilbertishilbertoflimit}
The Hilbert space of $L^2$ functions on $\overline{\Aa}^S$
can be obtained as a limit of Hilbert spaces:
\[
L^2(\overline{\Aa}^{\Ss})\cong \Ilim_{i\in I} L^2(\Aa_{\Gamma_i})
\mbox{ .}
\]
\end{proposition}
$L^2(\overline{\Aa}^{\Ss})$ is the kinematical Hilbert space
in Loop Quantum Gravity \cite{rovelli}.

\subsubsection{Embedded Graphs}
Let $M$ be a compact manifold of dimension $d$, and
 $\Gamma$ be a graph in $M$. More precisely, 
the   vertices $V_\Gamma$ are a set of points in $M$, and
the  (directed) edges $E_\Gamma$ are a set of non-self-intersecting
piecewise smooth curves in $M$ with starting and end points
given by the source map $\SO$ and range map $\RA$ respectively.

Let $M\times G$ be the 
trivial principal $G$-bundle
over $M$ with a fixed trivialization, where $G$ is a \emph{compact} Lie group.
Let $\Aa=\Omega^1(M,\g)$ be the space of smooth $G$-connections
on $M\times G$, with
gauge action of $C^\infty(M,G)$ by \[g \cdot A=
\operatorname{Ad}_g(A)+g d g^{-1} \mbox{ .}\]
\begin{definition}
\label{AGN:d:coarsegrainedhol}
\mbox{ }

\begin{enumerate}
\item
Let $\Gamma$ be a finite embedded graph in $M$.
Define the map 
\begin{equation}
\label{AGN:eqn:coarsegrainedhol}
\HOL_\Gamma:\Aa \to \Aa_\Gamma:=\HOM (\Hhg(\Gamma),G)
\end{equation}
to be the holonomy of $\nabla$ along the
path $\gamma \in\Hhg(\Gamma)$, where $\nabla\in \Aa$
 is a smooth $G$-connection.
\item
Let $\Ss=\{ \Gamma_i \}_{i\in I} $ be 
 a directed system of finite graphs in $M$.
Denote by $\HOL$ the map
\begin{equation}
\label{AGN:eqn:hol}
\HOL : \Aa \to \overline{\Aa}^\Ss
 :=\Plim \{ \Aa_{\Gamma_i}\}_{i\in I}
\end{equation}
induced from  the maps \eqref{AGN:eqn:coarsegrainedhol}.
\end{enumerate}
\end{definition}

\begin{proposition}[\cite{agn1}]
\label{AGN:prop:denserange}  
Let $G$ be a  connected compact Lie group,
\begin{enumerate}
\item 
For any finite
graph $\Gamma$, 
the map
$\HOL_\Gamma$ \eqref{AGN:eqn:coarsegrainedhol} is a surjection.
\item
Let $\Ss=\{\Gamma_i\}_{i\in I}$ be a directed system of finite graphs
in $M$. Then
$\Aa$ has a dense image in  $\overline{\Aa}^\Ss$
under the induced map $\HOL$ \eqref{AGN:eqn:hol}.
\end{enumerate}
\end{proposition}

From now on, we will assume that $G$ is \emph{connected}.

\begin{definition}
\label{AGN:d:densegraph}
A system of graphs $\Ss$ is said to be 
\textbf{densely embedded} in $M$
 if  for every
point $m\in M$ there exists a coordinate 
chart $x=(x_1,\ldots,x_d)$ around $m$ such that
for all open subset $U$ containing $m$
 in this coordinate chart there exists a
collection of edges $(e_1,\ldots,e_d)\subset U$
 belonging to graphs in $\Ss$ such that:
\begin{enumerate}
\item the $e_i$ 
are straight lines with respect to the coordinate chart,
\item the tangent vectors of the $e_i$ are linearly independent.
\end{enumerate}
\end{definition}

\begin{proposition}[\cite{agn1}]
\label{AGN:prop:embedding}
Given a densely embedded system of finite graphs $\Ss$ in $M$,
the map $\HOL$ injects $\Aa$ into  $\overline{\Aa}^\Ss$.
\end{proposition}

We give two examples of graph systems that
give rise to spaces of generalized connections
that $\Aa$ densely embeds into.

\begin{example}
\label{AGN:ex:triangulation}
Let $\Tt_1$ be a triangulation of $M$ and $\Gamma_1$ be the graph
consisting of all the edges in this triangulation with
any orientation.
Let $\Tt_{n+1}$ be the triangulation obtained by
barycentric subdividing each of the simplices in
 $\Tt_1$ $n$ times. The graph $\Gamma_{n+1}$ is
the graph consisting  the edges of $\Tt_{n+1}$ with any
orientation.
 In this way 
$\Ss_\triangle :=\{\Gamma_n \}_{n\in \naturalnumber}$
 is a directed system of finite graphs.
\end{example}

\begin{example}
\label{AGN:ex:dlattice}
 Let $\Gamma_1$ be a finite, 
$d$-dimensional lattice in $M$ and let $\Gamma_2$ be the
lattice obtained by subdividing each cell in $\Gamma_1$ into $2^d$ cells. 
Correspondingly, let $\Gamma_{n+1}$ be the lattice obtained by repeating $n$ such subdivisions
of $\Gamma_0$ . In this way 
 $\Ss_\square :=\{\Gamma_n\}_{n\in \naturalnumber}$
 is a directed system of finite graphs.
\end{example}

\subsection{The Limit of Spectral Triples on Graphs}
\label{AGN:dirac}
Given a finite graph $\Gamma$ and a
free generating set $\Ff_{\Gamma}$ for $\Gg(\Gamma)$, we saw how one builds
a spectral triple $(\Bb_\Gamma,
B(E\otimes \Hh_{\Ff_\Gamma}),\Dd_{\Ff_\Gamma})$ over it in Section~\ref{AGN:compactification}.
For systems of graph refinements like 
 $\Ss_\triangle$ and $\Ss_\square$ in Examples \ref{AGN:ex:triangulation} and \ref{AGN:ex:dlattice}, we
construct a compactification of the space
of connections.
With these examples in mind, we will restrict our directed set $I$
to the set of natural numbers $\naturalnumber$ 
and our graphs to be finite.
We would like to obtain a spectral triple over this 
compactified space of connections, and the way we proceed
is by taking the limit of some system of spectral triples
 \begin{equation}
\label{triplesystem}
(\Bb_{\Gamma_1},\Nn_{\Ff_{\Gamma_1}},\Dd_{\Ff_{\Gamma_1}}) 
\rightarrow 
(\Bb_{\Gamma_2},\Nn_{\Ff_{\Gamma_2}},\Dd_{\Ff_{\Gamma_2}}) 
\rightarrow 
(\Bb_{\Gamma_3},\Nn_{\Ff_{\Gamma_3}},\Dd_{\Ff_{\Gamma_3}}) 
\rightarrow \cdots
\mbox{ .}
\end{equation}

In the following, we will
construct the connecting morphisms
$(Q_{ij},P_{ij},\iota_{ij})$ (Definition~\ref{triplemorphism})
for the collection of 
spectral triples $\{ 
(\Bb_{\Gamma_j},\Nn_{\Gamma_j},\Dd_{\Gamma_j} )\}
_{j\in\naturalnumber}$
induced from a graph system 
$\{ \Gamma_j\}_{j\in\naturalnumber}$.

\subsubsection{Choice of Generators for $\Hhg(\Gamma)$}
\label{subsectioncoordinatechange}
It turns out that if we use the generating set $E_{\Gamma_i}$ for
each $\Gg(\Gamma_i)$,
 constructing the $\iota_{ij}$ 
intertwining the operators $\Dd_i$ and $\Dd_j$ is rather difficult.
AGN's solution is to introduce a ``change of generators''
\cite{agn1},
  which
 simplifies the work of constructing the $\iota_{ij}$.
 We will present
the choice 
of generators and 
 construct the $\iota_{ij}$ 
based on the new coordinates given to the connection spaces.

Recall that when $\Gamma_j$ is a refinement of $\Gamma_i$, we
obtain a morphism between the connection spaces
$\Aa_{\Gamma_i} \twoheadleftarrow \Aa_{\Gamma_j}$
induced by the groupoid embedding $\Hhg(\Gamma_i) \hookrightarrow 
\Hhg(\Gamma_j)$, and the
morphism consists of
  projections and multiplications
under
the identification \eqref{AGN:eqn:generatoridentification}.
One would like those connecting morphisms
 to be as simple as possible, and we have a procedure
to turn those connecting morphisms to be only composed of
projections.
We will demonstrate such a generator change procedure in the
following.
The identification of 
$\Aa_\Gamma = \HOM(\Hhg(\Gamma),G) 
\cong G^{\lvert E_\Gamma \rvert}$ 
\eqref{AGN:eqn:generatoridentification}
 is given 
by the $G$-assignments on 
 the free generating set, the set of edges.
However, one could use a different free generating set,
then $ \HOM(\Hhg(\Gamma),G)$ will be identified 
 with $G^{\lvert E_\Gamma \rvert}$ differently.
Following is  a description of the choice of
prefered free generating sets.
Let $\Ff_1:=E_{\Gamma_1}$, so
$\Ff_1$ freely generates 
$\Aa_{\Gamma_1}$.
Since $\Hhg( \Gamma_1 )\hookrightarrow 
\Hhg(\Gamma_2)$,
 we choose the 
generating set $\Ff_2$ of $\Hhg(\Gamma_2)$
 to be the set of images of $\Ff_1$ under
the groupoid  inclusion union with other paths in $\Hhg(\Gamma_2)$
 so that
$\Ff_2$ freely generates $\Hhg(\Gamma_2)$.
Similarly, choose the set $\Ff_3$ to be the set of
images of $\Ff_2$ under the groupoid inclusion $\Hhg(\Gamma_2)
\hookrightarrow \Hhg(\Gamma_3)$ union other paths in $\Hhg(\Gamma_3)$
so that $\Ff_3$ freely generates $\Hhg(\Gamma_3)$.
Repeat this procedure inductively.
Let us consider the following example:
\begin{eqnarray}
\label{edges}
\xymatrix{
\Gamma_1:&
 \ar@{-}[rrrrrrrr]_{e }|-{\object@{>}}  &  &   &&  &&&& \\
&&& &&\ar @{~>}[d]&& &&\\
&& &&&&&&\\
 \Gamma_2:& \ar@{-}[rrrr]_{e^{-}}|-{\object@{>}} &&&&
 \ar@{-}[rrrr]_{e^{+}}|-{\object@{>}} &&&&  \\
&&& &&\ar @{~>}[d]&& &&\\
&& &&&&&&\\
 \Gamma_3:&\ar@{-}[rr]_{e^{--}}|-{\object@{>}} &&
\ar@{-}[rr]_{e^{-+}}|-{\object@{>}} 
&&
 \ar@{-}[rr]_{e^{+-}}|-{\object@{>}} &&
 \ar@{-}[rr]_{e^{++}}|-{\object@{>}} 
 &&\\
&&&&& \vdots&& &&\\
 }  
\end{eqnarray}
First with
$\Ff_1 = E_{\Gamma_1} = \{e\}$,
the generating set $\Ff_2$ of $\Hhg(\Gamma_2)$ will consist
 the image of the edge $e\in E_{\Gamma_1}$, $(e^{-},e^{+})$, together
with another path, say $(e^{+})$.
Hence the generating set $\Ff_2$ for $\Hhg(\Gamma_2)$ is
$\Ff_2:= \{ (e^{-},e^{+}),(e^{+}) \}$, which
freely generates $\Hhg(\Gamma_2)$.
We repeat the process for $\Hhg(\Gamma_3)$.
Choose the generators in $\Hhg(\Gamma_3)$ to be
the image of $\Ff_2$, together with some other paths.
The image of $\{(e^{-},e^{+}),(e^{+})\}$ under the
groupoid inclusion
is $\{(e^{--},e^{-+},e^{+-},e^{++}),(e^{+-},e^{++})\}$, we
give $\Hhg(\Gamma_3)$ the set of free generators  
 $\Ff_3:=\{(e^{--},e^{-+},e^{+-},e^{++}),(e^{+-},e^{++})
,(e^{-+}),(e^{++})\}$.

The ultimate goal of this generator choice is that
we want the set of free generators for
 $\Hhg(\Gamma_{i+1})$
to be the set of free generators for $\Hhg(\Gamma_{i})$
together with some other paths, so that we 
obtain a nested sequence of generating sets
$\Ff_1\subset \Ff_2 \subset \ldots$
for the system $\{\Hhg(\Gamma_i)\}_{i\in \naturalnumber}$.
In this process of choosing the generators,
there are choices to make for the generators 
that do not come from images of the previous 
generating set, for instance we could have
chosen the free generating set of $\Hhg(\Gamma_3)$
to be 
\[\{(e^{--},e^{-+},e^{+-},e^{++}),(e^{+-},e^{++})
,(e^{-+}),(e^{-+},e^{+-},e^{++})\} .\]
We allow such a freedom here.
However, according to AGN there
is a physics reason
for choosing   
\[\{(e^{--},e^{-+},e^{+-},e^{++}),(e^{+-},e^{++})
,(e^{-+}),(e^{++})\}\]
instead,
which is discussed in \cite{agn3}.


After obtaining the nested sequence of 
generating sets for 
$\left\{ \Hhg(\Gamma_i)_{i\in \naturalnumber} \right\}$,
we identify each $\HOM(\Hhg(\Gamma_i),G)$ with $\MAP(\Ff_i,G)$,
which gives $\Aa_{\Gamma_i}\cong G^{\lvert E_{\Gamma_i} \rvert }$
a new identification. Hence, the projective system 
$\left\{ \Aa_{\Gamma_i} \right\}_{i\in\naturalnumber}$ becomes
$\left\{ G^{\lvert E_{\Gamma_i} \rvert} 
\right\}_{i\in\naturalnumber}$
with the connecting morphisms
\[
G^{\lvert E_{\Gamma_i} \rvert} 
 \twoheadleftarrow
 G^{ \lvert E_{\Gamma_{i+1}} \rvert } 
\]
given by \emph{projections} of dropping 
$\lvert E_{\Gamma_{i+1}} \rvert -\lvert E_{\Gamma_{i}} \rvert$
coordinates.
For convenience, we will assume that
the connecting morphisms drop the 
\emph{last} $\lvert E_{\Gamma_{i+1}} \rvert - 
\lvert E_{\Gamma_{i}} \rvert$ coordinates.

\subsubsection{Morphisms between Spectral Triples on Graphs}

Given two finite graphs $\Gamma$ and $\Gamma'$, where 
$\Gamma'$ is a refinement of $\Gamma$, we will now construct
a morphism $(Q,P,\iota)$ between their 
corresponding spectral triples
 $(\Bb_\Gamma,B(E\otimes\Hh_{\Ff_\Gamma}),\Dd_{\Ff_\Gamma})$
and $(\Bb_{\Gamma'},
B(E\otimes \Hh_{\Ff_{\Gamma'}}),\Dd_{\Ff_{\Gamma'}})$,
where the generating set $\Ff_{\Gamma'}$ is chosen
according to Subsection~\ref{subsectioncoordinatechange}.

The surjection $\Aa_{\Gamma}\twoheadleftarrow \Aa_{\Gamma'}$
 of connection spaces induces
 an embedding of 
$L^2$ functions
$L^2(\Aa_\Gamma) \hookrightarrow 
L^2(\Aa_{\Gamma'} )$.
Under the new identification 
$\HOM(\Hhg(\Gamma),G)\cong G^{\lvert E_\Gamma \rvert}$
 described in
Subsection~\ref{subsectioncoordinatechange},
the map 
\begin{equation}
\label{L2embedding}
L^2(G^{\lvert E_{\Gamma} \rvert } ) \hookrightarrow 
L^2(G^{\lvert E_{\Gamma'} \rvert } )
\end{equation}
is given by extending the functions by constant in the 
last $\lvert E_{\Gamma'} \rvert - \lvert E_{\Gamma} \rvert$
variables.
We then construct the map
\begin{equation}
\label{Sembedding}
 S^{\otimes \lvert E_{\Gamma} \rvert }
\hookrightarrow S^{\otimes \lvert E_{\Gamma'} \rvert } \mbox{ ,}
\end{equation}
given by embedding $S^{\otimes \lvert E_{\Gamma} \rvert }$
as the subspace $S^{\otimes \lvert E_{\Gamma} \rvert }\otimes 
 \mathbb{I}^{\otimes
 { \lvert E_{\Gamma'} \rvert -\lvert E_{\Gamma} \rvert }} \mbox{ ,}
$
where $S$ is the cyclic representation of the $C^*$-algebra
$\CL(\g)$ with (normalized) cyclic vector $\mathbb{I}$.
The product of maps \eqref{L2embedding} and 
\eqref{Sembedding} gives us the desired Hilbert space map 
\[
\iota:
\Hh_{\Ff_\Gamma} \hookrightarrow \Hh_{\Ff_{\Gamma'}} \mbox{ .}
\]

Let $\iota_* :\Ww(\g^{\lvert E_{\Gamma} \rvert } )
\to \Ww(\g^{\lvert E_{\Gamma'} \rvert } )$
be the  map induced by  embedding 
$\g^{\lvert E_{\Gamma} \rvert }$ 
as the \emph{first} $\lvert E_{\Gamma} \rvert $
coordinates in
$ \g^{\lvert E_{\Gamma'} \rvert}$.
Since such an embedding 
$\g^{\lvert E_{\Gamma} \rvert } \hookrightarrow 
\g^{\lvert E_{\Gamma'} \rvert }$ preserves both 
the \emph{metrics} and 
the \emph{Lie brackets}, $\iota_*$ is an algebra homomorphism.
$\Ww(\g^{\lvert E_\Gamma \rvert})$ acts 
on $\Hh_{\Gamma}$ as discussed in Subsection~\ref{DH}.

 \begin{lemma} 
\label{quantumweilcompatible}
Let $w\in \Ww(\g^{\lvert E_\Gamma \rvert})$
and $\iota \in \Hh_{\Gamma}$.
Then $\iota_* w ( \iota (\eta) ) = \iota (w (\eta))$ whenever
 $w(\eta)$ is defined.
\end{lemma}
\begin{proof}
 $\g^{\lvert E_{\Gamma} \rvert}$
differentiates in the first $\lvert E_{\Gamma} \rvert$
variables on $L^2(G^{\lvert E_{\Gamma'} \rvert},\MATRIX)$
and acts by Clifford action on the first
$\lvert E_{\Gamma} \rvert$ copies of
 $S^{\otimes \lvert E_{\Gamma'} \rvert}$, so 
the action of $\Ww(g^{\lvert E_{\Gamma} \rvert})$
on $\Hh_{\Ff_{\Gamma'}}$ preserves the subspace
$\Hh_{\Ff_{\Gamma}}$. The result follows.
\end{proof}

The surjection 
$\HOM(\Hhg(\Gamma),G) \twoheadleftarrow 
\HOM(\Hhg(\Gamma'),G)$
 induces an embedding of 
$G$-valued
maps 
\begin{equation}
\label{someembedding}
\MAP\left(\HOM(\Hhg(\Gamma),G),G\right) \hookrightarrow
\MAP\left(\HOM(\Hhg(\Gamma'),G),G \right) \mbox{ .}\end{equation}
Fix a vertex $\nu$ in $\Gamma$, 
one has 
an embedding 
$\Hhg_\nu(\Gamma)\hookrightarrow \Hhg_\nu(\Gamma')$
of the isotropy groups.
The group homomorphism $h$ \eqref{embedisotropy}
completes
the following commutative diagram:
\begin{equation}
\label{Bdiagram}
\xymatrix{
\Hhg_\nu(\Gamma) 
\ar[dd] \ar[rrr]^{h\hspace{1cm}} &&& \MAP\left(\HOM(\Hhg(\Gamma),G),G\right)  \ar[dd] \\
&& \ar@(ul,dl) & \\
\Hhg_\nu(\Gamma') \ar[rrr]^{h\hspace{1cm}} &&&
\MAP\left(\HOM(\Hhg(\Gamma'),G),G\right) 
} \mbox{ .}
\end{equation}
In particular, when \eqref{someembedding} is restricted to
the $G$-valued smooth functions \newline
$C^\infty(\Aa_\Gamma,G)\subset 
\MAP\left(\HOM(\Hhg(\Gamma),G),G\right) $.
Therefore, diagram \eqref{Bdiagram} restricts to the following
diagram:
\[
\xymatrix{
\Hhg_\nu(\Gamma) 
\ar[dd] \ar[rrr]^{h} &&&
C^\infty(\Aa_{\Gamma},G)  \ar[dd] \\
&& \ar@(ul,dl) & \\
\Hhg_\nu(\Gamma') \ar[rrr]^{h} &&&
C^\infty(\Aa_{\Gamma'},G) 
} \mbox{ .}
\]

We let $Q:\Bb_\Gamma \to \Bb_{\Gamma'}$ 
be the $*$-homomorphism
extended from the group homomorphism
$h(\Hhg_\nu(\Gamma)) \to h(\Hhg_\nu(\Gamma'))$.

Let $P:B(E\otimes \Hh_{\Ff_\Gamma}) \to 
B(E\otimes \Hh_{\Ff_{\Gamma'}})$ be 
the $*$-homomorphism
induced from $\iota:\Hh_{\Ff_\Gamma} \hookrightarrow \Hh_{\Ff_{\Gamma'}}$.
We can think of $B(E\otimes \Hh_{\Ff_\Gamma})$ as block matrices in 
$B(E\otimes \Hh_{\Ff_{\Gamma'}})$.

\begin{proposition}
\label{graphtriplemorphism}
$(Q,P,\iota )$ is a morphism between the spectral triples
$(\Bb_{\Gamma},B(E\otimes \Hh_{\Ff_\Gamma}),\Dd_{\Ff_\Gamma})$
and $(\Bb_{\Gamma'},B(E\otimes \Hh_{\Ff_{\Gamma'}}),
\Dd_{\Ff_{\Gamma'}})$.
\end{proposition}
\begin{proof}
The operator $\Dd_{\Ff_{\Gamma'}} \in \Ww(\g^{\lvert E_{\Gamma'} 
\rvert})$ \eqref{DiraconGamma}
 can be written as
\[
\Dd_{\Ff_{\Gamma'}}:= 
\sum _{\gamma \in \Ff_{\Gamma'}}
 D_\gamma
=  
\sum _{\gamma \in \Ff_{\Gamma}}
 D_\gamma
+\sum _{\gamma \in \Ff_{\Gamma'}\backslash \Ff_{\Gamma}}
 D_\gamma
 \]
Furthermore, 
$\sum _{\gamma \in \Ff_{\Gamma'}\backslash \Ff_{\Gamma}}
 D_\gamma$
 acts on 
$\iota(\Hh_{\Ff_\Gamma})$ by zero and 
$\sum _{\gamma \in \Ff_{\Gamma}} 
 D_\gamma = \Dd_{\Ff_\Gamma}$
Thus, together with Lemma~\ref{quantumweilcompatible},
\[
\iota(\Dd_{\Ff_{\Gamma}}\eta) =  (\Dd_{\Ff_{\Gamma}})( \iota(\eta))
=\Dd_{\Gamma'}(\iota(\eta))
\]
for $\eta \in \Hh_{\Ff_{\Gamma}}$
 and the diagram
\[
\xymatrix{
\DOM(\Dd_{\Gamma}) \ar[dd]^{\iota}
\ar[rr]^{\Dd_{\Gamma}} && \Hh_{\Gamma}
\ar[dd]^{\iota}\\
&&\\
\DOM(\Dd_{\Gamma'})\ar[rr]^{\Dd_{\Gamma'}} && \Hh_{\Gamma'}
}
\]
commutes.

As the surjection $\Aa_{\Gamma} \twoheadleftarrow 
\Aa_{\Gamma'}$ induces the injections
\begin{eqnarray*}
C^\infty(\Aa_\Gamma,\MATRIX) & \hookrightarrow& 
C^\infty(\Aa_{\Gamma'},\MATRIX)
\\
  L^2(\Aa_\Gamma, E) 
&\hookrightarrow&  L^2(\Aa_{\Gamma'} , E) \mbox{ .}
\end{eqnarray*}
$C^\infty(\Aa_\Gamma,\MATRIX)$ as a sub-algebra of 
$C^\infty(\Aa_{\Gamma'},\MATRIX)$
preserves the subspace\newline $L^2(\Aa_\Gamma,E)
\subset L^2(\Aa_{\Gamma'},E)$.
Since $\Bb_{\Gamma}$ is a sub-algebra of 
$C^\infty(\Aa_\Gamma,\MATRIX)$ and it acts trivially on 
$S^{\otimes \lvert E_{\Gamma} \rvert}$, it preserves
the subspace $E\otimes \Hh_{\Ff_\Gamma}\subset 
E\otimes\Hh_{\Ff_{\Gamma'}}$.
Finally, by construction, $B(E\otimes \Hh_{\Ff_\Gamma})$ as
a sub-algebra of $B(E\otimes \Hh_{\Ff_{\Gamma'}})$ preserves 
the subspace $\Hh_{\Ff_\Gamma}\subset \Hh_{\Ff_{\Gamma'}}$ as well.
Hence $(Q,P,\iota)$ forms a morphism as defined in Definition~\ref{triplemorphism}
\end{proof}


\subsubsection{The System of Spectral Triples}
Given a system of finite graphs $\{\Gamma_j\}_{j\in \naturalnumber}$
and a system of generating sets $\{\Ff_j\}_{j\in\naturalnumber}$
described in Subsection~\ref{subsectioncoordinatechange},
there is associated a collection of spectral triples
$\{(\Bb_{\Gamma_j},B(E\otimes\Hh_{\Ff_j}),\Dd_{\Ff_j}\}_{j\in 
\naturalnumber}$, and between every pair of spectral triples
in the collection, there is a morphism $(Q_{ij},P_{ij},\iota_{ij})$
between them. In fact this collection of morphisms 
make the collection of spectral triples into a system of spectral triples.

\begin{proposition}
\label{graphtriplesystem}
The collection of spectral triples
$\{(\Bb_{\Gamma_j},B(E\otimes\Hh_{\Ff_j}),\Dd_{\Ff_j})\}_{j\in 
\naturalnumber}$
together with the collection of morphisms
 $\{(Q_{ij},P_{ij},\iota_{ij}\}_{i\leq j}$
forms an inductive system of spectral triples.
\end{proposition}
\begin{proof}
By construction,
the diagram
\[
\xymatrix{
(\Bb_{\Gamma_i},B(E\otimes \Hh_{\Ff_i}),\Dd_{\Ff_i})
\ar[rr]^{(Q_{ij},P_{ij},\iota_{ij})} 
\ar[rdd]_{(Q_{ik},P_{ik},\iota_{ik})}    &&
(\Bb_{\Gamma_j},B(E\otimes\Hh_{\Ff_j}),\Dd_{\Ff_j}) \ar[ldd]^{(Q_{jk},P_{jk},\iota_{jk})} \\
& 
&&\\
& (\Bb_{\Gamma_k},B(E\otimes \Hh_{\Ff_k}),\Dd_{\Ff_k})
}
\]
commutes for $i\leq j \leq k$. Hence, 
$\{(\Bb_{\Gamma_j},B(E\otimes \Hh_{\Ff_j}),\Dd_{\Ff_j})\}_{j\in 
\naturalnumber}$ 
is an inductive system of spectral triples.
\end{proof}

Denote the limit triple of $\{(\Bb_{\Gamma_j},B(E\otimes \Hh_{\Ff_j}),\Dd_{\Ff_j})\}_{j\in 
\naturalnumber}$
by 
$
(\Bb,B(E\otimes \Hh),\Dd)
$. That is, $\Bb:=\Ilim \Bb_{\Gamma_j}$, 
$\Hh:=\Ilim \Hh_{\Ff_j}$, and $\Dd:=\Ilim \Dd_{\Ff_j}$.

\begin{remark}
The inductive system of Lie algebras 
$\{\g^{\lvert E_{\Gamma_j} \rvert } \}_{j\in\naturalnumber}$
with connecting morphisms given by 
embedding into $\g^{\lvert E_{\Gamma_j} \rvert }$
into the first $\lvert E_{\Gamma_j} \rvert$
coordinates of $\g^{\lvert E_{\Gamma_{j+1}} \rvert } $,
gives rise to an inductive system of 
quantum Weil algebras $
\{\Ww( \g^{\lvert E_{\Gamma_j} \rvert }) \}
_{j\in\naturalnumber}$.
Denote its limit $\Ilim \Ww( \g^{\lvert E_{\Gamma_j} \rvert }) $
by $\Ww$, then $\Dd$ is an element of $\Ww$. 
The limit quantum Weil algebra $\Ww$ can be
viewed as the limit of the universal enveloping algebras
$\Ilim \Uu(   \g^{\lvert E_{\Gamma_j} \rvert }) $
tensor with the limit of the Clifford algebras
$\Ilim \CL(   \g^{\lvert E_{\Gamma_j} \rvert }) $,
 which are
the Canonical Commutation Relation (CCR) algebra, 
and the Canonical Anti-Commutation Relation (CAR) algebra.
\end{remark}

\begin{remark}
We have assumed our directed set $I$ for the
graph system $\{\Gamma_i\}_{i\in I}$ to be the set of 
natural numbers $\naturalnumber$. Up to now, our construction
works for any countable directed set $I$, as long as we
the system of generating sets $\{\Ff_i\}_{i\in I}$ for the
system of groupoids $\{\Gg(\Gamma_i)\}_{i\in I}$ satisfy 
the relation that $\Ff_i \subset \Ff_j$ whenever
$i<j \in I$.
\end{remark}

Unfortunately, $(\Bb,B(E\otimes \Hh),\Dd)$ is not a
semi-finite spectral triple. The reason being that
$\Dd$ does not have compact resolvent.

The operator $\Dd_{\Ff_j}\in \Ww(\g^{\lvert 
E_{\Gamma_j} \rvert })$ \eqref{DiraconGamma}
can be written as
\begin{equation}
\label{eq:partialsumnotation}
\Dd_{\Ff_j}:=\sum_{k=1}^{j} \Dd_k
\end{equation}
where
\begin{equation}
\label{newnotation}
\Dd_k:= \sum _{\gamma\in \Ff_k \backslash \Ff_{k-1}} D_\gamma
\end{equation}
is the Dirac operator 
corresponds to those edges added to $\Gamma_{k-1}$
to form $\Gamma_k$ via the refinement.

We fix a real valued sequence $\left\{a_k\right\}_{k=1}^\infty$,
and scale each $\Dd_k$ by $a_k$ so that
\[
\Dd_{\Ff_j}=\sum_{k=1}^{j}
a_k \Dd_k \mbox{ .}
\]
Propositions \ref{graphtriplemorphism} and
\ref{graphtriplesystem} continuous to hold as long as
$a_k\neq 0$ for all $k$, so
the limit triple $(\Bb,\Nn,\Dd)$ is still
defined, however it now depends on the sequence $\{a_k\}$,
as $\Dd$ depends on $\{a_k\}$.
  Therefore,  the sequence $\{a_k\}$ will be included as
a dynamical variable in $\Dd$.

The physical interpretation of the sequence $\{a_k\}$
 is the following \cite{agn3}.
In the new notations described in Equations~\eqref{eq:partialsumnotation}
and~\eqref{newnotation},
the operator $\Dd_{\Ff_{i+1}}$ (viewed as an operator on $\Hh$)
is obtained from
$\Dd_{\Ff_i}$ by adding on an extra part $\Dd_{i+1}$
that corresponds to edges obtained from a graph refinement.
Those new edges are necessarily shorter. As
the Dirac operator carries a notion of the 
inverse volume in the language of noncommutative geometry
\cite{greenbook},
the part $\Dd_{i+1}$ that is added to $\Dd_{\Ff_i}$
to make $\Dd_{\Ff_{i+1}}$ is supposed to carry less
weight for it to correspond to those shorter edges.
 For graph refinements in Example~\ref{AGN:ex:triangulation}, there is not
an obvious choice of a lesser weight that one should
assign to $\Dd_{i+1}$. However,
in Example~\ref{AGN:ex:dlattice}, the refinement is 
obtained from subdividing each cell in a $d$-lattice into 
$2^d$ cells. Thus, the ``volume'' the new
edges from the refinement carry should be $\frac{1}{2^d}$
of the original. As the relation between the Dirac operator and
volume is an inversed  relation, the operator $\Dd_{i+1}$
ought to be scaled to $2^{di} \Dd_{i+1}$.
Therefore, when the lattice graphs are used, one obtains
a  weight assignment sequence $\{2^{dk}\}_{k=1}^\infty$
that depends only on the dimension of the base manifold.
From now, we will be using the property that  the 
directed set being
$\naturalnumber$.


\begin{proposition}
\label{dimSis1}
If $\dim (S) =1$, then $(\Bb,B(E\otimes\Hh),\Dd)$ is a spectral triple
whenever $a_k^2 \nearrow \infty$.
\begin{proof}
With Proposition~\ref{propertiesoflimittriple}, we only 
need to prove that $\Dd$ has compact resolvent.
That is, when the projection
$1_{[0,\lambda]} (\Dd^2)$ has finite rank for all 
$\lambda < \infty$, where $1_{[0,\lambda]}$ is the
characteristic function supported on the
closed interval $[0,\lambda]\subset \mathbb{R}$.
First we look at non-zero eigenvalues of $\Dd^2$.
Since
$a_k^2 \nearrow \infty$, every non-zero eigenvalue
of the operator $a_k^2 \Dd_k^2$
will eventually be outside of the interval $[0,\lambda]$
as $k$ increases.
 As  $\Dd^2=\sum_{k=1}^{\infty} a_k^2 \Dd_k^2$, 
there are only finitely many eigenvalues (not counting multiplicity)
 of $\Dd^2$ in $[0,\lambda]$.
Each of such non-zero eigenvalue of $\Dd^2$ in $[0,\lambda]$
 must have finite multiplicity because 
the non-zero eigenvalues of each $a_k^2\Dd_k^2$ have 
finite multiplicities.
Now it remains to show that the zero eigenspace of $\Dd^2$ has finite
dimension.
As a consequence of Theorem~\ref{AGN:thm:eckhard},
 the zero eigenspace of $\Dd^2$ 
is $E\otimes \Ilim_{i\in\naturalnumber} S^{\otimes \lvert E_{\Gamma_i}
\rvert}$,
 where $E\subset L^2(\overline{\Aa}^{\Ss}, E)$ is embedded
as constant functions,
 which has rank $\dim(E)$
 as $\dim(S)=1$. Hence $1_{[0,\lambda]} (\Dd^2)$
has finite rank as long as $E$ is a finite dimensional 
representation of $\MATRIX$ and $\Dd$ has compact resolvent.
\end{proof}
\end{proposition}


\subsection{Semi-finite Spectral Triple}
\label{AGN:semifinite}

\subsubsection{Limit of Semi-finite Spectral Triples}

When $\dim(S)>1$, Proposition~\ref{dimSis1} no longer
holds, as $1_{\{0\}}(\Dd^2)=E\otimes \Ilim_{i\in\naturalnumber} S^{\otimes \lvert E_{\Gamma_i}
\rvert}$ has infinite rank.
Fortunately, there is semi-finite trace from
a von Neumann algebra other than $B(E\otimes \Hh)$ that
computes the ``dimension'' of $S$ to be $1$, and 
it in the end gives us a semi-finite spectral triple
by reducing the complication to the case of 
Proposition~\ref{dimSis1}.
We go back to the system of spectral triples
$\{(\Bb_{\Gamma_j},B(E\otimes \Hh_{\Ff_j}),\Dd_{\Ff_j})\}_{j\in 
\naturalnumber}$, and replace each $B(E\otimes \Hh_{\Ff_j})$
with a more suitable semi-finite von Neumann algebra 
$\Nn_{\Ff_j}\subset B(E\otimes \Hh_{\Ff_j})$.

Recall that  Theorem~\ref{AGN:thm:eckhard} asserts that
 $D\in \Uu(\g)\otimes \CL(\g)$, as an operator on $L^2(G)\otimes S$,
 has kernel $\ker(D)=\mathbb{C}\otimes S$. The projection onto
$\ker(D)$ is $P_1\otimes 1 \in B(L^2(G))\otimes B(S)$,
where $P_1$ denotes the projection onto the space 
of constant functions.
Denote the weak operator closure of $B(L^2(G))\otimes \CL(\g)$
by $\Nn_1$, it comes equipped with the semi-finite trace
 $\tau_1$ by extending
 $\TR\otimes \TR_{\CL}$, where $\TR$ is the operator trace on 
$B(L^2(G))$ and $\TR_{\CL}$ is the Clifford trace on $\CL(\g)$.
Then clearly $D\in\Uu(g)\otimes \CL(\g)$ is affiliated
with $\Nn_1$ and 
\[ \tau_1(P_1\otimes 1)=\TR(P_1)\cdot \TR_{\CL}(1) = 1 \mbox{ .}
\]
Therefore, $\ker(D)$ has ``dimension'' $1$ relative to the
trace $\tau_1$.
In fact $\tau_1$ is the operator trace on
 $B(L^2(G)\otimes S)$ normalized by dividing the 
dimension of $S$. Thus each eigenvalue of $D$ now has 
 ``multiplicity'' $\frac{1}{\dim(S)}$ of before.


On the Hilbert space $E\otimes\Hh_{\Ff_j}$, let
$\Nn_{\Ff_j}$ be the weak operator closure of
$B(L^2(G^{\lvert E_{\Gamma_j} \rvert },E))
\otimes \CL (\g^{\lvert E_{\Gamma_j} \rvert })$, 
which is equipped with the trace $\tau_{\Ff_j}$ by 
extending $\TR\otimes \TR_{\CL}$, where again
$\TR$ is the operator trace and $\TR_{\CL}$ is the
Clifford trace on $\CL(\g^{\lvert E_{\Gamma_j} \rvert})$.

\begin{lemma}
$(\Bb_{\Gamma_j},\Nn_{\Ff_j},\Dd_{\Ff_j})$ forms
a semi-finite spectral triple.
\end{lemma}
\begin{proof}
$\Dd_{\Ff_j}\in \Ww(\g^{\lvert E_{\Gamma_j} \rvert})$ 
is affiliated with $\Nn_{\Ff_j}$ and has compact resolvent 
relative to $\tau_{\Ff_j}$.
The rest of the proof proceeds similar to the proof of
Proposition~\ref{graphtriple}.
\end{proof}

For $i\leq j$, the map of embedding $\g^{\lvert E_{\Gamma_i} \rvert}$
to the first $\lvert E_{\Gamma_i} \rvert$ copies
in $\g^{\lvert E_{\Gamma_j} \rvert}$ induces 
an algebra map $\CL( \g^{\lvert E_{\Gamma_i} \rvert})
\to \CL(\g^{\lvert E_{\Gamma_j} \rvert})$.
The map of extending 
functions in $L^2(G^{\lvert E_{\Gamma_i} \rvert},E)$ by 
constants in the last $\lvert E_{\Gamma_j} \rvert
-\lvert E_{\Gamma_i} \rvert$ variables defines
a Hilbert space embedding \newline
$L^2(G^{\lvert E_{\Gamma_i} \rvert},E)
\hookrightarrow L^2(G^{\lvert E_{\Gamma_j} \rvert},E)$,
which induces an embedding   \newline
 $B(L^2(G^{\lvert E_{\Gamma_i} \rvert},E)) 
\hookrightarrow B(L^2(G^{\lvert E_{\Gamma_i} \rvert},E))$.
By extending the product of these two maps, we
obtain a   $*$-homomorphism
\begin{eqnarray}
\label{vNmap}
P_{ij}: \Nn_{\Gamma_i} \to \Nn_{\Gamma_j} \mbox{ .}
\end{eqnarray}

\begin{lemma}
$\{(\Bb_{\Gamma_j},\Nn_{\Ff_j},\Dd_{\Ff_j})\}_{j\in 
\naturalnumber}$ forms an inductive system of semi-finite spectral triples
\end{lemma}
\begin{proof}
By construction, $P_{ij}(a)(\iota(\eta)) = \iota(a(\eta))$
for $a\in\Nn_{\Ff_i}$ and $\eta \in \Hh_{\Ff_i}$.
The rest of the proof proceeds similar to the proof of 
Proposition~\ref{graphtriplesystem}.
\end{proof}

Let $(\Bb,\Nn,\Dd)$ be the limit of 
$\{(\Bb_{\Gamma_j},\Nn_{\Ff_j},\Dd_{\Ff_j})\}_{j\in 
\naturalnumber}$, i.e., $\Nn$ is the von Neumann algebra
limit $\Ilim \Nn_{\Ff_j}$.

By extending 
Proposition \ref{limitofhilbertishilbertoflimit},
one has $ \Ilim L^2(G ^{\lvert E_{\Gamma_j} \rvert }, E)=
L^2(\overline{\Aa}^{\Ss},E)$.

$\Nn$ is the weak operator closure of 
$B(L^2(\overline{\Aa}^{\Ss} , E))
\otimes \Ilim \CL(\g^{\lvert E_{\Gamma_j} \rvert})$, 
it is equipped with the semi-finite trace $\tau$ by extending
$\TR\otimes \TR_{\CL}$, where 
$\TR$ is the operator trace on \newline
$B(L^2(\overline{\Aa}^{\Ss}, E))$
and $\TR_{\CL}$ is the Clifford trace 
on $ \Ilim \CL(\g^{\lvert E_{\Gamma_j} \rvert})$.

\begin{remark}
The algebra $\Nn$ is of Type $\operatorname{II_\infty}$, 
as the CAR algebra 
$\Ilim \CL(\g^{\lvert E_{\Gamma_j} \rvert})$,
being an ``infinite tensor'' of Type $\operatorname{I_n}$ algebras,
is Type $\operatorname{II_1}$.
\end{remark}

\begin{theorem}
$(\Bb,\Nn,\Dd)$ is a semi-finite spectral triple 
whenever $a_k^2 \nearrow \infty$.
\end{theorem}
\begin{proof}
With Proposition~\ref{propertiesoflimittriple}, we only 
need to prove that $\Dd$ has compact resolvent.
That is, 
$\tau( 1_{[0,\lambda]} (\Dd^2))<\infty$ for all 
$\lambda < \infty$, where $1_{[0,\lambda]}$ is the
characteristic function supported on the
closed interval $[0,\lambda]\subset \mathbb{R}$.

Since
$a_k^2 \nearrow \infty$, every non-zero eigenvalue
of the operator $a_k^2 \Dd_k^2$
will eventually be outside of the interval $[0,\lambda]$
as $k$ increases. As  $\Dd^2=\sum_{k=1}^{\infty} a_k^2 \Dd_k^2$,
there are only finitely many eigenvalues 
 of $\Dd^2$ in $[0,\lambda]$.
Each such non-zero eigenvalue  in $[0,\lambda]$
 has finite multiplicity relative to $\tau$ because 
the non-zero eigenvalues of each $a_k^2\Dd_k^2$ have 
finite multiplicities.
By Theorem~\ref{AGN:thm:eckhard},
$1_{\{0\}}(\Dd^2)=P_E\otimes 1$,
where $P_E$ is the projection onto the space
of constant sections $E$ in $L^2(\overline{\Aa}^\Ss,E)$,
so
 \[\tau(1_{\{0\}}(\Dd^2))=\TR(P_E) \cdot \TR_{\CL}(1)=\dim(E)
\mbox{ .}\]
 Hence $\tau(1_{[0,\lambda]} (\Dd^2))<\infty$
for all $\lambda$ and $\Dd$ has compact resolvent.
$(\Bb,\Nn,\Dd)$ is a semi-finite spectral triple.
\end{proof}

The upshot of using the semi-finite von Neumann algebra
$\Nn$ is that, ordinarily 
\[\TR( 1_{\{0\}}(\Dd^2) ) = \dim(E)\cdot \dim \left( 
\Ilim S^{\otimes \lvert E_{\Gamma_j} \rvert }  \mbox{ .}
\right) = \dim(E) \cdot \infty\]
One hopes to get a finite number out of it, so 
we divide the above quantity by $\infty$, more
precisely by the dimension of
$\Ilim S^{\otimes \lvert E_{\Gamma_j} \rvert }$.
Thus, we form the quantity 
\[
\frac{1}
{  \dim\left( S^{\otimes \lvert E_{\Gamma_j}\rvert }\right)} 
\TR_{\Ff_j}( 1_{\{0\}}(\Dd_{\Ff_j}^2) ) \mbox{ ,}
\]
where $\TR_{\Ff_j}$ is the operator trace on $B(E\otimes \Hh_{\Ff_j})$.
By taking the limit as $j\to\infty$, we get $\dim(E)$, which is finite.

This procedure has a flavor of renormalization in Quantum field theory.

\subsubsection{The $\mathbb{Z}_2$-grading}

Suppose that the cyclic representation $S$ of $\CL(\g)$ is
the Clifford  algebra $\CL(\g)$ itself, which is 
a $\mathbb{Z}_2$ graded vector space with say $+1$-eigenspace
$\CL^+(\g)$ and $-1$-eigenspace $\CL^-(\g)$.
Then it equips the Hilbert space $E\otimes\Hh_{\Ff_\Gamma}$ with
a grading so that the action of the algebra $\Bb_\Gamma$ is even and
the action of the Dirac operator $\Dd_{\Ff_\Gamma}$ is odd, making
the spectral triple $(\Bb_\Gamma,B(E\otimes \Hh_{\Ff_\Gamma}),
\Dd_{\Ff_\Gamma})$ even.
But it does not make our semi-finite spectral triple
$(\Bb_\Gamma,\Nn_{\Ff_\Gamma},\Dd_{\Ff_\Gamma})$ even, since
the grading operator is \emph{not} an element of
the von Neumann algebra $\Nn_{\Ff_\Gamma}$.
However, if $\g$ is even dimensional of dimensional $2k$, one could equip 
$S$ with  the $\mathbb{Z}_2$ grading given by the chirality element
$(\sqrt{-1})^{k} e_1\cdots e_{2k}$
in $\CL(\g)$ for the basis $\{e_i\}_{i=1}^{2k} \subset \g$. The
chirality element is self-adjoint and squares to $1$. It
induces a grading operator $\chi_{\Ff_\Gamma}$
of $E\otimes \Hh_{\Ff_\Gamma}$ in $\Nn_{\Ff_\Gamma}$ so that
$\chi_{\Ff_\Gamma}$ anti-commutes with $\Dd_{\Ff_\Gamma}$ and commutes with
any element in $\Bb_{\Gamma}$. As a result, 
$(\Bb_\Gamma,\Nn_{\Ff_\Gamma},\Dd_{\Ff_\Gamma})$
is an \emph{even} semi-finite spectral triple with respect to
$\chi_{\Ff_\Gamma}$
 if $G$ is even dimensional.

With notations as before, we let 
$(\Bb_{\Gamma_j},\Nn_{\Ff_j},\Dd_{\Ff_j})_{j\in I}$ be a system
of semi-finite spectral triples. Suppose that $G$ is 
\emph{even dimensional},
then $(\Bb_{\Gamma_j},\Nn_{\Ff_j},\Dd_{\Ff_j})$
 is an even semi-finite spectral triple with respect to the
grading operator $\chi_{\Ff_j}\in \Nn_{\Ff_j}$.
The map $P_{j\infty}:\Nn_{\Ff_j}\to \Nn$ sends each 
$\chi_{\Ff_j}\in \Nn_{\Ff_j}$ to $\Nn$, so we obtain 
a net of operators $\{\chi_{\Ff_j}\in \Nn_{\Ff_j}\}_{j\in I}$
in $\Nn$. Denote by $\chi$ its strong operator limit.
Then by Theorem~\ref{propertiesoflimittriple},
$(\Bb,\Nn,\Dd)$ is an \emph{even} semi-finite spectral triple
with respect to $\chi\in\Nn$.

\begin{proposition}
$(\Bb,\Nn,\Dd)$ is an even (resp. odd) semi-finite
spectral triple if the Lie group is even (resp. odd)
dimensional.
\end{proposition}


\section{JLO Theory}

\label{ch:jlo}

JLO theory due to Jaffe, 
Lesniewski and Osterwalder \cite{jlo}
 is a cohomological Chern character that assigns a cocycle,
hence a class, in entire cyclic
cohomology to a weakly $\theta$-summable 
semi-finite spectral triple.
The cohomology class 
is homotopy invariant, thus the JLO
character descends
to a map from $\KKEI$-homology classes to 
entire cyclic cohomology classes.
The cocycle computes the Type II Fredholm index 
and spectral flow of the operator $\Dd$, hence
it provides an index formula in the setting of noncommutative
geometry.
Here we will give a summary of the JLO theory. For details
please refer to \cite{l}.

\subsection{Entire Cyclic Cohomology}
Let $A$ be a unital Banach algebra over $\mathbb{C}$  and
$C^n(A)$ to be the space of
$n$-linear functionals over $A$.
Define the operators
 $b: C^n(A) \rightarrow  C^{n+1}(A)$ and $B: C^n(A)\rightarrow  C^{n-1}(\Bb)$  by the formulas
\begin{eqnarray*}
 (b\phi_n)\left(a_0,\ldots,a_n\right)&:=&\sum^{n-1}_{j=0} (-1)^{j} \phi_n(a_0,\ldots,a_{j}a_{j+1},\ldots,a_n) 
+ (-1)^{n}\phi_n(a_n a_0,a_1,\ldots,a_{n-1}) \mbox{ ,}\\
(B\phi_n)\left(a_0,\ldots,a_n\right)&:=&\sum^n_{j=0}(-1)^{nj}
\phi_n(1,a_j,\ldots,a_n,a_0,\ldots,a_{j-1}) \mbox{ ,}
\end{eqnarray*}
for $\phi\in C^n(\Bb)$.

Simple calculation shows that $b^2=B^2=Bb+bB=0$.
Therefore $(b+B)^2=0$ and we get the following bicomplex:
\[
    {
\xygraph{
!{<0cm,0cm>;<2cm,0cm>:<0cm,2cm>::}
!{(1,5) }*+{\vdots}="15"
!{(2,5) }*+{\vdots}="25"
!{(3,5) }*+{\vdots}="35"
!{(4,5) }*+{\vdots}="45"
!{(0,4) }*+{\cdots}="04"
!{(1,4) }*+{C^3(A)}="14"
!{(2,4) }*+{C^{2}(A)}="24"
!{(3,4) }*+{C^{1}(A)}="34"
!{(4,4) }*+{C^{0}(A)}="44"
"15":"14"^{B}
"25":"24"^{B}
"35":"34"^{B}
"45":"44"^{B}
"14":"04"^{b}
"24":"14"^{b}
"34":"24"^{b}
"44":"34"^{b}
!{(0,3) }*+{\cdots}="03"
!{(1,3) }*+{C^2(A)}="13"
!{(2,3) }*+{C^{1}(A)}="23"
!{(3,3) }*+{C^{0}(A)}="33"
"14":"13"^{B}
"24":"23"^{B}
"34":"33"^{B}
"13":"03"^{b}
"23":"13"^{b}
"33":"23"^{b}
!{(0,2) }*+{\cdots}="02"
!{(1,2) }*+{C^1(A)}="12"
!{(2,2) }*+{C^{0}(A)}="22"
"13":"12"^{B}
"23":"22"^{B}
"12":"02"^{b}
"22":"12"^{b}
!{(0,1) }*+{\cdots}="01"
!{(1,1) }*+{C^0(A)}="11"
"12":"11"^{B}
"11":"01"^{b}
!{(4,3) }*+{}="a"
!{(2,1) }*+{}="b"
"a":"b"_{(b+B)}
}
}\mbox{ .}
\]
The space $C^\bullet(A):=\prod_{n=0}^\infty C^{n}(A)$  has a natural
$\mathbb{Z}_2$ grading given by 
 $C^+(A)=\prod_{k=0}^\infty C^{2k}(A)$ and
 $C^-(A)=\prod_{k=0}^\infty C^{2k+1}(A)$.
We get a cochain complex $\left(C^\bullet(A), b+B \right)$ with the odd boundary map
$b+B$. However, the cohomology of this cochain complex is trivial.
In order to make it nontrivial, we 
restrict our attention to cochains 
that satisfy the following growth condition as $n$ varies.

\begin{definition}
\label{defnentire}
Define
\[
 C^\bullet _\omega (A):=\left\{ \phi _\bullet \in C^\bullet (A): 
\sum_{n=0}^\infty  \Gamma(\frac{n}{2}) \left\lVert \phi_{n}\right\rVert
z^{n}
   \mbox{ is an entire function in $z$ }  \right\}
\]
where $\left\lVert\phi_n\right\rVert:=\sup
\{|\phi_n(a_0,\ldots,a_n)|: \left\lVert a_j\right\rVert \leq 1  \mbox{ }\forall j\}$.
We call the cochain $\phi_\bullet$ 
\textbf{entire} if
$\phi_\bullet \in C^\bullet _\omega (A)$. It is easy to see that if 
$\phi_\bullet$ is entire, so is $(b+B)\phi_\bullet$. Hence
 $\left(C^\bullet_\omega(A),b+B\right)$ is a subcomplex of $\left(C^\bullet(A), b+B \right)$.
The cohomology defined by $\left(C^\bullet _\omega(A),b+B\right)$ is the \textbf{entire cyclic cohomology} of $A$, denoted 
$\HE^\bullet(A)=\HE^+(A)\oplus \HE^-(A)$. 
\end{definition}


\subsection{The JLO Character}
\label{s:jlochar}
The JLO character  assigns cocycles in entire cyclic cohomology to 
semi-finite spectral triples satisfying an appropriate
 summability condition. We begin by defining the summability conditions of 
main concern.

\begin{definition}
 An semi-finite spectral triple $\KCYCLE$ is:
\begin{itemize}
\item[(a)]
\textbf{$p$-summable} if $\tau\left((1+\Dd^2)^{-p/2}\right) <\infty$ ;
\item[(b)]
\textbf{$\theta$-summable} if $
 \tau(e^{-t\Dd^2}) < \infty $ for all $t>0$;
\item[(c)]
\textbf{weakly $\theta$-summable} 
if $ \tau(e^{-t\Dd^2}) < \infty $ for some $0<t<1$.

\end{itemize}
\end{definition}

Observe that 
$p$-summability implies $\theta$-summability,
which in turn implies weak $\theta$-summability.

Let $\Delta_n:=\{(t_1,\ldots,t_n)\in \mathbb{R}^n : 0\leq t_1 \leq \cdots \leq t_n \leq 1\}$ be the
standard $n$-simplex and $d^nt=dt_1\cdots dt_n$ is the standard Lesbeque measure on $\Delta_n$ with
volume $\frac{1}{n!}$.
\begin{definition}
\label{JLO}
\mbox{ }
\begin{enumerate}
\item 
 The  \textbf{odd JLO character}
 $\Ch^-(\Dd)\in C^-( A )$ 
of a weakly $\theta$-summable \textit{odd} 
semi-finite spectral triple
$\KCYCLE$ is defined to be

\[
 \Ch^-(\Dd):=\sum_{k=0}^\infty \Ch^{2k+1}(\Dd)\mbox{ ,}
\]
\item
The \textbf{even JLO character} $\Ch^+(\Dd)\in C^+(A)$
of a weakly $\theta$-summable \textit{even} 
semi-finite spectral triple
 $\KCYCLE$ is defined to be
\[
 \Ch^+(\Dd):=\sum_{k=0}^\infty \Ch^{2k}(\Dd)\mbox{ ,}
\]
where
$A$ is the closure of $\Bb$ with respect to the norm
$\lVert \cdot \rVert + \lVert [\Dd,\cdot] \rVert$ and 

 \[
  \left(\Ch^{n}(\Dd),(a_0,\ldots,a_{n})_n\right):=
\int_{\Delta_n}\tau 
\left( \chi 
a_0 e^{-t_1 \mathcal{D}^2 }
[\Dd,a_1] e^{-(t_2-t_1) \Dd^2 }
\ldots
[\Dd,a_n] e^{-(1-t_n) \Dd^2 } 
\right)d^nt
 \]
with the convention $\chi=1$ if 
 $n$
 is odd.
\end{enumerate}
\end{definition}


\begin{theorem}[\cite{l}]
\label{jloisacocycle}
The JLO character $\Ch^\bullet(\Dd)$
 is an entire cyclic cocycle in $\HE^\bullet( A)$.\newline
More specifically,
\[
 \Ch^\bullet(\Dd)\in C^\bullet _\omega( A)\hspace{0.3cm} \mbox{ and }\hspace{0.3cm}
(b+B)\Ch^\bullet(\Dd)=0\mbox{ .}
\]

\end{theorem}
As a result, the JLO character defines an entire cyclic
 cohomology class called the JLO class.





\subsection{Homotopy Invariance of the JLO Class}
\label{homotopyinvarianceJLO}

Suppose that $\Dd_t$ is a $t$-parameter family 
of operators so that it
defines a differentiable family of  weakly $\theta$-summable
semi-finite spectral triples $\left(\Bb,\Nn,\Dd_t  \right)$.
Namely, $\Dd_t$  is a $t$-parameter family of self-adjoint operators on $\Hh$ 
with common domain of definition
so that the following is satisfied:
\begin{itemize}
\item
$\Dd_t$ is  affiliated with $\Nn$ for all $t$,
\item
For all
 $a\in \Bb$, $[\Dd_t , a]$ is 
a  norm-differentiable family of operators in  $  \Nn$  ,
 and 
there is a constant $C$ for each compact interval such that
$ \left\lVert [\Dd_t,a]\right\rVert \leq C \left\lVert a\right\rVert$,
\item
$(1+\Dd_t^2)^{-1/2}$ is a  norm-differentiable
 family of operators in $\Kk$,
\item
There exists a $u\in(0,1)$ such that
$\tau(e^{-u\Dd_t^2})$ is  bounded for each compact interval.
 \end{itemize}

If $(\Bb,\Nn,\Dd_t)$ is equipped with a $\mathbb{Z}_2$ grading 
$\chi\in\Nn$ so that $a$ is even for all $a\in \Bb$ and
$\Dd_t$ is odd for all $t$, then similarly we call the
family of semi-finite spectral triple $(\Bb,\Nn,\Dd_t)$ \emph{even}.
\begin{theorem}[\cite{l}]
\label{cobound}
 If $\dot{\Dd_t}=F_t|\Dd_t|^{1+\varepsilon}+R_t$
 for $0\leq \varepsilon <1$ and 
 $F_t, R_t \in \Nn$ are continuous families of operators 
that are uniformly bounded in $t$ then
the entire cyclic cohomology class of
$\Ch^\bullet(\Dd_t)$ is independent of $t$.

\end{theorem}


The following Proposition gives a stability of bounded
perturbation of weakly $\theta$-summable
semi-finite spectral triple.
It is Theorem C in \cite{getzlereven}.
\begin{proposition}[\cite{l}]
\label{thmC}
 For a 
weakly  $\theta$-summable semi-finite spectral triple
 $\KCYCLE$, and an operator $V\in\Nn$ 
such that $V$ has the same degree as $\Dd$, i.e. 
$\lvert V \rvert _\chi = \lvert \Dd \rvert _\chi$.
 Then  $\left(\Bb, \Nn,\Dd+V \right)$ is
again a weakly $\theta$-summable semi-finite spectral triple
\end{proposition}

\section{The JLO Class of AGN's Space of Connections}
\label{ch:jloandagn}

When the Dirac  operator $\Dd$ of 
the noncommutative space of connections of AGN
 is weakly $\theta$-summable,
there is an entire cyclic cocycle
associated to it.
As we have observed in Section~\ref{AGN:construction}, AGN's
spectral triple, and
 hence
the associated JLO cocycle, depends on the weight
assignment, which is a diverging sequence $\{a_k\}$.
We give an explicit condition on allowable
perturbations of the given weight assignment so that
the associated JLO class remains invariant.
In a more recent paper \cite{agn3},
 Aastrup-Grimstup-Paschke-Nest eliminate the
 weight ambiguity by using lattice graphs,
which results in the most reasonable choice of 
weight assignment
that depends only on the dimension of the base manifold.
If one re-runs AGN's construction on a sub-manifold,  
the resulting spectral triple will be
defined using the weight assignment corresponding to the sub-manifold,
which in general would be 
 different from the
spectral triple obtained from pulling back the
construction on the  full manifold.
Although the two 
semi-finite spectral triples
 are assigned different weights,
Section~\ref{AGN:jlo} eliminates that
dimension dependence at the level of 
entire cyclic cohomology by exploiting
the  homotopy invariance of the JLO character described
 in Section~\ref{homotopyinvarianceJLO} 
(see \cite{l} for details).
Section~\ref{AGN:theta} analyzes
 the weak $\theta$-summability of the Dirac operator
$\Dd$ in terms of the diverging sequence $\{a_k\}$.

\subsection{Weight Independence of the JLO Class}
\label{AGN:jlo}
Let $\{\Gamma_k\}_{k\in I}$ be a directed system of finite graphs
with the system of groupoid generators $\{\Ff_{k}\}_{k\in I}$ chosen
according to Subsection~\ref{subsectioncoordinatechange}.
Recall that
\begin{itemize}
\item 
 $D_\gamma$ is the basic Dirac operator 
associated to a path in the a graph (see \eqref{AGN:eqn:basicdirac});
\item
$\Dd_k:=\sum _{\gamma \in \Ff_k \backslash \Ff_{k-1}} D_\gamma$
 is the Dirac  operator corresponds to those edges added  
to $\Gamma_{k-1}$ to form $\Gamma_k$ (see \eqref{newnotation});
\item
$\Dd_{\Ff_n}=\sum_{k=1}^n a_k \Dd_k$
is the Dirac operator on the finite graph $\Gamma_n$ (see \eqref{eq:partialsumnotation}).
\end{itemize}
By Theorem~\ref{propertiesoflimittriple}(2),
 $\Dd$ is the strong resolvent
 limit of $\{ \Dd_{\Ff_n}\}_{n=1}^\infty$.
 We will write the derivative of the weight changes of each term
in $\Dd$ as the operator $F_t \lvert \Dd _t\rvert ^{1+\varepsilon}
+R_t$,
and apply
Theorem~\ref{cobound} to give
a concrete condition on allowable variation of the weight 
assignments $\{a_k\}_{k=1}^\infty$. In particular,
we will show that
the variation corresponding to the dimension change in the base
manifold gives rise to a transgression cochain that is entire.

\begin{lemma}
\label{AGN:lemma:sequencecontrol}
Let $\{a_k(t)\}$ be a family of sequences parametrized by $t$ so that
each $a_k(t)$ is differentiable in $t$.
Let $\Dd(t)$ be
the $t$-family of Dirac operators  given by the weight assignment
$\{a_k(t)\}_{k=1}^\infty$ (i.e.
$\sum_{k=1}^n a_k(t) \Dd_k = \Dd_{\Ff_n} (t)\to \Dd(t)$ in 
the strong resolvent sense).
If there exist $\varepsilon \in [0,1)$ and $m\in(0,\infty)$
 such that
\[
\sup _k \left( \frac{ \lvert\dot{a_k}(t)\rvert }{
\lvert a_k(t)\rvert^{1+\varepsilon} } \right)
 \leq m  \mbox{ }\mbox{ for all } t \mbox{ ,}
\]
 then   there exist continuous families of
operators $F_t,R_t\in \Nn$   such that
\[\dot{\Dd}(t)=F_t \lvert \Dd(t)\rvert ^{1+\varepsilon} +R_t
 \mbox{ .}\]
and $\lVert F_t \rVert, \lVert R_t\rVert \leq m$ uniformly.
\end{lemma}
\begin{proof}
We  first analyze the operator
\[
\lvert \dot{\Dd}(t) \rvert \left(
1 + \lvert \Dd(t) \rvert ^{1+\varepsilon} \right)^{-1} \mbox{ .}
\]
Since $\Dd_{\Ff_n}(t) \to \Dd(t)$ in the strong resolvent sense for all $t$,
\[
\frac{\Dd_{\Ff_n}(t+h) - \Dd_{\Ff_n}(t)}{h} \longrightarrow
\frac{ \Dd(t+h) - \Dd(t)}{h}
\]
for all $h$, and we have the strong resolvent limit
\[
\sum_{k=1}^n \dot{a_k}(t) \Dd_k=
\dot{\Dd_{\Ff_n}}(t) \longrightarrow \dot{\Dd}(t) \mbox{ .}
\]
Thus, we obtain the limit
\[
\lvert \dot{\Dd}_{\Ff_n}(t) \rvert \left(
1 + \lvert \Dd_{\Ff_n}(t) \rvert ^{1+\varepsilon} \right)^{-1} 
 \longrightarrow 
\lvert \dot{\Dd}(t) \rvert \left(
1 + \lvert \Dd(t) \rvert ^{1+\varepsilon} \right)^{-1} \mbox{ ,}
\]
which a priori is only in  the strong resolvent sense.

Let $0<\frac{1}{c}$ be the smallest \emph{non-zero} eigenvalue of $D^2$.
For instance, $c=8$ when $G=SU(2)$ by Remark~\ref{remark:smallesteigenvalue}. Then
the open interval $(0,1)$ does not contain any part of 
the spectrum of $c\Dd_k^2$ as each $\Dd_k^2$ is a finite 
sum of $D^2$'s.
%
We bound the operator
$\lvert \dot{\Dd}_{\Ff_n}(t) \rvert \left(
1 + \lvert \Dd_{\Ff_n}(t) \rvert ^{1+\varepsilon} \right)^{-1} $ by using
multi-variable functional calculus on the   
set of commuting self-adjoint operators
 $\{c \Dd_k^2\}_{k=1}^n$.
The following multi-varible function is bounded
when $x_k\geq 1$ or $=0$ for all $k$:
\begin{eqnarray*}
f_t^n\left(\{x_k\}_{k=1}^n \right)
:=\frac{c^{\varepsilon} \left( \sum_{k=1}^n \dot{a_k}(t)^2 x_{k}
 \right)^{\frac{1}{2}}}
{  c^{1+\varepsilon}+ \left( \sum_{k=1}^n a_k(t)^2 x_{k}\right) 
^{\frac{1+\varepsilon}{2}}   } \leq
\frac{ c^{\varepsilon}   m
\left( \sum_{k=1}^n a_k(t)^{2(1+\varepsilon)} x_{k}
 \right)^{\frac{1}{2}}}{
  c^{1+\varepsilon}+ \left( \sum_{k=1}^n a_j(t)^2 x_{k}\right) 
^{\frac{1+\varepsilon}{2}}   } \hspace{2cm}\\
\leq
\frac{c^{\varepsilon}  m
\left( \sum_{j=1}^n a_k(t)^{2(1+\varepsilon)} x_{k}^
{(1+\varepsilon)}
 \right)^{\frac{1}{2}}}{
  c^{1+\varepsilon}+ \left( \sum_{k=1}^n a_k(t)^2 x_{k}\right) 
^{\frac{1+\varepsilon}{2}}  }
\leq
\frac{c^{\varepsilon}  m
\left( \sum_{j=1}^n a_k(t)^{2 } x_{k}
 \right)^{\frac{1+\varepsilon}{2}}}{
  c^{1+\varepsilon}+ \left( \sum_{k=1}^n a_k(t)^2 x_{k}\right) 
^{\frac{1+\varepsilon}{2}}   } \leq c^{\varepsilon}  m \mbox{ .}
\end{eqnarray*}
As $\lVert f_t^n \rVert _\infty \leq c^{\varepsilon} m$,
$f_t^n\left(\{c\Dd_k^2\}_{k=1}^\infty
\right)$ is uniformly bounded.
By construction
\[
f_t^n\left(\{c\Dd_k^2\}_{k=1}^\infty
\right) = 
\lvert \dot{\Dd}_{\Ff_n}(t) \rvert \left(
1 + \lvert \Dd_{\Ff_n}(t) \rvert ^{1+\varepsilon} \right)^{-1} \mbox{ ,}
\]
so $\lvert \dot{\Dd}_{\Ff_n}(t) \rvert \left(
1 + \lvert \Dd_{\Ff_n}(t) \rvert ^{1+\varepsilon} \right)^{-1} $
is uniformly bounded by $ c^{\varepsilon}  m$ and 
\[
\lvert \dot{\Dd}_{\Ff_n}(t) \rvert \left(
1 + \lvert \Dd_{\Ff_n}(t) \rvert ^{1+\varepsilon} \right)^{-1} 
 \longrightarrow 
\lvert \dot{\Dd}(t) \rvert \left(
1 + \lvert \Dd(t) \rvert ^{1+\varepsilon} \right)^{-1} 
\] in strong operator topology.
The fact that
$\Dd_k$ is affiliated with $\Nn$ for each $k$ implies that 
  $\lvert \dot{\Dd}_{\Ff_n}(t) \rvert \left(
1 + \lvert \Dd_{\Ff_n}(t) \rvert ^{1+\varepsilon} \right)^{-1} 
\in \Nn$ for each $n$.
As $\Nn$ is strong operator closed,
  the limit
$ \lvert 
\dot{\Dd}(t) \rvert \left(
1 + \lvert \Dd(t) \rvert ^{1+\varepsilon} \right)^{-1}
 \in \Nn$.
Set $F_t=R_t=\dot{\Dd} (t) \left(
1 + \lvert \Dd(t) \rvert ^{1+\varepsilon} \right)^{-1}
 \in \Nn$, then
$ \dot{\Dd } (t) = 
F_t \lvert \Dd(t) \rvert ^{1+\varepsilon}    + R_t$.
Furthermore,
observe that $\lVert f_t^n \rVert _\infty $
uniformly bounded by $ c^{\varepsilon}m$ 
for all $n$,
the operator norms of $F_t$ and $R_t$ are also uniformly
bounded by $c^{\varepsilon} m$.
\end{proof}
\begin{theorem}
\label{AGN:theorem:final}
If $\{a_k(t)\}$ is a differentiable family of sequences so that
$\Dd(t)$ is weakly $\theta$-summable for all $t$, and that
there exist $\varepsilon \in [0,1)$ and $m\in(0,\infty)$
 such that
\[
\sup_k \left( \frac{\lvert \dot{a_k}(t) \rvert }
{\lvert a_k(t) \rvert ^{1+\varepsilon}} \right)
 \leq m \mbox{ } \mbox{ for all } t \mbox{ ,}
\]
 then
$\Ch(\Dd(t))$ defines the same entire cyclic cohomology class 
in $\HE ^\bullet (A)$
for any $t$, where $A$ is the closure of $\Bb$ with 
respect to the norm $\lVert b \rVert _{\operatorname{Lip}}
:=\lVert b \rVert + \lVert [\Dd_{t_0},b] \rVert $ for a fixed
 $t_0$ and all $b\in\Bb$.
\end{theorem}
\begin{proof}
By Lemma~\ref{AGN:lemma:sequencecontrol}, 
$\dot{\Dd}(t)=A_t \lvert \Dd(t) \rvert ^{1+\varepsilon}
 +R_t $.
By Theorem~\ref{cobound},
$\hCh^\bullet(\Dd(t),\dot{\Dd(t)})$ is entire and
$[\Ch(\Dd(t_1))]=[\Ch(\Dd(t_2))]$ for
any $t_1$, $t_2$ as an entire cyclic cohomology class,
 and the proof is complete.
\end{proof}

Now we will give an application of Theorem~\ref{AGN:theorem:final}.
Recall that $M$ is the base manifold of dimension $d$,
and $\{\Gamma_i\}_{i\in\naturalnumber}$ is a system of embedded graphs
in $M$.
In \cite{agn3}, Aastrup-Grimstrup-Nest choose
the sequence $a_j=(2^d)^j$
by using lattice graphs of dimension $d$.
Let $a_j(t)=(2^{d+t})^j$, then we see that

\begin{equation}
\label{AGN:eqn:latticebound}
\sup_j \left( \frac{\lvert \dot{a_j}(t)\rvert }{
\lvert a_j(t) \rvert ^{1+\varepsilon}} \right)
 = \sup_j \left( 
 \frac{  j \ln (2)    }{(2^{(d+t)j})^{\varepsilon}}
 \right) \mbox{ ,}
\end{equation}
which is finite as long as $t+d >0$
and  $\varepsilon > 0$.
By Theorem~\ref{AGN:theorem:final}, we see that
as long as $t$ is so that $\Dd(t)$ is weakly $\theta$-summable, 
which can be achieved by an overall scaling of $\{a_j\}$ (see Section~\ref{AGN:theta}),
 the choice of $d$ does not affect its JLO class.
We will use this observation to show that the 
entire cyclic cohomology class associated to AGN's Dirac operator
is well-behaved under immersions.

Let $N$ and $M$ be compact manifolds of dimension $c$ and $d$ 
respectively.
We give $N$  a system of $c$-lattice graphs
 (see Example~\ref{AGN:ex:dlattice}) 
and 
assume that $N$ immerses in $M$ in a way that
the 
system of  $c$-lattice graphs in
$N$ extends to a system of  $d$-lattice graphs in $M$.
That is, the $d$-lattice graphs are \emph{refinements} of 
the $c$-lattice graphs such that,
other than the joining vertices,
the new edges added for the refinements lie entire outside of
$N$,
 Fix a vertex $\nu$ in 
a $c$-lattice graph of $N$, denote by
$\left(\Bb_c,\Nn_c, \Dd_c\right)$ and
$\left(\Bb_d,\Nn_d, \Dd_d\right)$
the semi-finite spectral triples 
constructed according to Section~\ref{AGN:construction}
 for $N$ and $M$.
Because the $d$-lattice graphs are refinements of 
$c$-lattice graphs, we get an embedding of the
algebras $\iota:\Bb_c \hookrightarrow \Bb_d$. 
The embedding $\iota$ pulls back the semi-finite spectral triple 
$\left(\Bb_d,\Nn_d, \Dd_d\right)$  
to  $\left(\Bb_c ,\Nn_d, \Dd_d\right)$.
Denote by $A$ the closure of $\Bb_c$ under
the norm  $\lVert b \rVert _{\operatorname{Lip}}:= \lVert b \rVert + 
\lVert [\Dd_c,b]\rVert$ for $b\in\Bb_c$.

\begin{theorem}
\label{AGN:cor:immersion}
When
 $\left(\Bb_c ,\Nn_c, \Dd_c\right)$  and
 $\left(\Bb_c,\Nn_d, \Dd_d\right)$  are weakly
 $\theta$-summable, 
their JLO cocycles 
define the same JLO class in 
$\HE^\bullet(A)$.
\end{theorem}
\begin{proof}
The unit of $A_c$ acts as a projection $p$ on $\Hh_d$ and decomposes
$\Hh_d$ into $p\Hh_d  \oplus (1-p) \Hh_d $.
Since the $d$-lattices are refinements
of the $c$-lattices, 
by the generator choice described in
 Subsection~\ref{subsectioncoordinatechange}, the 
Hilbert space $p\Hh_d$ is precisely $\Hh_c$ and the
representation $\rho_d \circ \iota$ is just $\rho_c$.
Also $\Dd_d$ decomposes  into $\Dd_d ' +\Dd_d ''$ with 
$\Dd_d'$ affiliated with $p\Nn_d p=\Nn_c$ and $\Dd_d ''$
affiliated with $(1-p) \Nn_d (1-p)$.
As a result, $\Ch^\bullet(\Dd_d) = \Ch^\bullet (\Dd_d')
+\Ch^\bullet (\Dd_d '')$, and $\Ch^\bullet (\Dd_d '')=0\in 
\HE^\bullet(A)$.
 Now by construction, the
 Dirac operators $\Dd_c$ and $ \Dd_d ' $
differ only by the defining sequence $\{a_j\}$.
By Theorem~\ref{AGN:theorem:final} and Equation~\eqref{AGN:eqn:latticebound},
 $\Ch(\Dd_c)$ equals $\Ch( \Dd_d')$
 up to an entire
 coboundary, hence 
$\Ch(\Dd_c)$ and $\Ch(\Dd_d)$
 define the same class in $\HE^\bullet(A_c)$.
\end{proof}

\subsection{Weak $\theta$-summability}
\label{AGN:theta}

In \cite{agn0}, the term
$ \tau(e^{-\Dd^2}) $ has the form of a 
formal Feynman path
integral, thus weak
$\theta$-summability has a physical motivation
of the path integral being finite.
One would like to know for what sequences 
$\{a_j\}$, $\Dd$ is weakly $\theta$-summable.

We give a characterization of the weakly
$\theta$-summable condition of $\Dd$ in terms
of the sequence $\{a_j\}$.
However, the analysis will necessarily depend
on the graph system that is deployed.
The most general graph system is given by adding 
one edge at a time, because under our choice generators
described in Subsection~\ref{subsectioncoordinatechange},
any refinement of a graph 
can be obtained by successively adding edges one at a time.
As a result, our graph system gives
the following system of connection spaces:
\[
G\twoheadleftarrow G\times G 
\twoheadleftarrow G\times G \times G
\twoheadleftarrow \cdots \twoheadleftarrow G^{n-1}
\twoheadleftarrow G^n
\twoheadleftarrow \cdots \mbox{ .}
\]
Therefore, $\Dd=\sum^\infty_{k=1} a_k \Dd_k$,
where
$\Dd_k$ is just $D$ \eqref{AGN:eqn:basicdirac} acting
on the $k$-th copy of $G$.

Recall that we extended the action of $\Dd$ on
 $\Hh:=\Ilim \Hh_{\Ff_j}$ to
$E\otimes \Hh$ by letting it act as the identity on $E$.
Since $E$ is finite dimensional, 
the summability of $\Dd$ remains unchanged whether $\Dd$
acts on $E$ or not. Hence, we will for simplicity let
$\Dd$ act on $\Hh$ instead of $E\otimes \Hh$.

Let 
 $\tau_1$ be the semi-finite trace
on the weak operator
 closure of $  B(L^2(G)) \otimes  \CL (\g)$
  extended from 
 $ \TR\otimes \TR_{\CL} $, where   $\TR$
 denotes
the   operator trace on  $B(L^2(G))$
and $\TR_{\CL}$ denotes the
Clifford trace
 (note that $B(L^2(G)) \otimes  \CL (\g)$
is a Type-$\operatorname{I_\infty}$ algebra).
Since the Lie group $G$ is assumed to be finite
dimensional, the   Dirac operator $D=\frac{1}{\sqrt{-1}}\sum_{i=1}^q e_i \otimes e_i$ 
acting on $L^2(G)\otimes S$
is
$p$-summable for some finite $p$ greater
than the dimension of $G$, hence $D$ is also
 $p$-summable with respect to $\tau_1$. 
Moreover, there exists a smallest non-zero eigenvalue 
 for
$D^2$, so $ 
\left(D ^2+1\right)^p ( \left. D 
\right \rvert_{\KER D^\perp} )^{-2p} $
is a bounded operator.
In the case when $G=SU(2)$, the smallest non-zero eigenvalue is
$\lvert \rho \rvert^2 = \dim (\g)/24=\frac{1}{8}$, so
 $ 
\left(D ^2+1\right)^p ( \left. D 
\right \rvert_{\KER D^\perp} )^{-2p} $ is bounded by 
$65^p$. As a result
\begin{eqnarray*}
\sum _{\lambda\in \sigma(D)\backslash \{0\} }
\lvert\lambda\rvert^{-p}
\tau_1\left(E_{D}(\lambda) \right)
&=&\tau_1\left( \left\lvert D_{\ker D^\perp } \right\rvert
^{-p} \right) \\
&\leq &
\left\lVert 
\left(D ^2+1\right)^p ( \left. D 
\right \rvert_{\KER D^\perp} )^{-2p} \right\rVert  \cdot
\tau_1 \left(    \left(D ^2+1\right)^{-p}       \right) 
<\infty \mbox{ ,}
\end{eqnarray*}
where $\sigma(D)$ is the spectrum of $D$
 and the $E_D(\lambda)$
are the  spectral projections of $D$
onto the $\lambda$-eigenspace.
Note that $\tau_1(E_D(0))=1$ (see Theorem~\ref{AGN:thm:eckhard}). 
We can index the eigenvalues of $D$ by integers
counting multiplicities with respect to $\tau_1$,
so that $\lambda_n\in \sigma(D)$,
$\ldots \leq \lambda_{-2}\leq \lambda_{-1}
\leq \lambda_{0} \leq \lambda_{1} \leq \lambda_{2} 
\leq \ldots$, and $\lambda_n=0$ only when $n=0$.
Recall that the multiplicity of an eigenvalue $\lambda$ is given 
by the number $\tau_1\left( E_D (\lambda)\right)$.

%

One would like the trace 
\[
\tau(e^{-u\Dd^2})=\tau(e^{-u\sum_j a_j^2\Dd_j^2})
=\sum_{n\in \Zz ^\infty}
 e^{-u\sum_j a_j^2\lambda_{n_j}^2} 
\]
to be finite,
where
$ \Zz^\infty $ is the space of 
functions from $\naturalnumber \ni j \mapsto n_j \in \Zz$ of finite support.
We break the sum on the right  according to the cardinality of the
support. Namely, it is broken into sums over
 functions 
with  \textit{no support},
supported on  \textit{one} point,
 and supported on \textit{two} points etc.
Thus, 
\[\sum_{n\in \Zz ^\infty}
e^{-u\sum_j a_j^2\lambda_{n_j}^2 }=\sum_{i=0}^\infty 
\sum_{\lvert \SUPP(n)\rvert = i }
e^{-u\sum_j a_j^2\lambda_{n_j}^2} \mbox{ .}\]
 We compute:
\newline
Sequences $n\in\Zz^\infty$ with $\lvert \SUPP(n) \rvert =0$:
\[
\sum_{n\in \Zz^\infty , \lvert \SUPP(n) \rvert =0 } 
e^{-u\sum_j a_j^2\lambda_j^2} = e^{-0}=1 \mbox{ .}
\]
Sequences $n\in\Zz^\infty$
 with $\lvert \SUPP(n) \rvert =1$:
\begin{eqnarray*}
\sum_{n \in \Zz^\infty , \lvert \SUPP(n) \rvert =1 } 
e^{-u\sum_j a_j^2\lambda_n^2} &=& \sum_{j}
\sum_{\lambda \in \sigma(D) \backslash
 \{ 0 \}}
 e^{-u a_j^2 \lambda^2} \tau_1(E_D(\lambda))
\\ &\leq&
\sum_{j}
\lVert e^{- u a_j^2 x^2}\lvert x \rvert ^p
\rVert_\infty \left(
\sum_{\lambda\in \sigma(D) \backslash
 \{ 0 \}}
 \lvert \lambda \rvert ^{-p} \tau_1(E_D(\lambda))\right)
\\
&= & \left( \sum_j a_j^{-p} \right)
\left(\frac{p}{2e u}\right)^\frac{p}{2}
\tau_1\left( \left\lvert D_{\ker D^\perp } \right\rvert
^{-p} \right)
 =: Y
\mbox{ ,}
\end{eqnarray*}
where $\lVert \cdot \rVert _\infty $ is the uniform norm. \newline
Sequences $n\in\Zz^\infty$ with $\lvert \SUPP(n) \rvert =2$:
\begin{eqnarray*}
\sum_{n\in \Zz^\infty,\lvert \SUPP(n) \rvert =2  } 
e^{-u \sum_j a_j^2 \lambda_{n_j}^2} &=& 
\sum_{i\neq j\in\naturalnumber}
\sum_{(x,y)\in \sigma(D)^2 \backslash
  (*,0)\cup (0,*)  }
 e^{-u(a_i^2 x^2 +a_j^2 y^2) } 
\tau_1\left( E_D(x)\right) \tau_1 \left( E_D(y) \right)\\
&\leq &
Y^2
\mbox{ .}
\end{eqnarray*}
With an induction argument, we obtain 
\[
\sum_{\lvert \SUPP(n)\rvert = i }
e^{-u\sum_j a_j^2\lambda_{n_j}^2} \leq Y^i \mbox{ .}\] Thus
\[\sum_{n\in \Zz ^\infty}
e^{-u\sum_j a_j^2\lambda_{n_j}^2 }\leq \sum_{i=0}^\infty 
Y^i \mbox{ .}\]
We conclude that $\tau(e^{-u\Dd^2})<\infty$ if
the geometric series
$
\sum _{i=0}^\infty Y^i $ converges.
In other words, when 
\[
Y:= \left( \sum_j a_j^{-p} \right)
\left( 
\left(\frac{p}{2eu}\right)^\frac{p}{2}
\tau_1 \left( \left\lvert D_{\ker D^\perp } 
\right\rvert^{-p} \right) \right)<1 \mbox{ ,}
\]
 or the $p$-norm of the reciprocal sequence
 $\left\lVert \left\{ \frac{1}{a_j}\right\}\right\rVert_p$ is
less than
 $\left(\sqrt{\frac{p}{2eu}}
\left\lVert 
 \left\lvert D_{\ker D^\perp } \right\rvert^{-1}
\right\rVert_p\right)^{-1}$.
The condition
\begin{eqnarray}
\label{e:thetacondition}
\left\lVert \left\{ \frac{1}{a_j} \right\}\right\rVert _p
 < \sqrt{ \frac{2eu}{p}}
\left\lVert 
\left\lvert D_{\ker D^\perp } \right\rvert
^{-1}
\right\rVert_p ^{-1} \end{eqnarray}
can be interpreted as that for any fixed $u$, any diverging
sequence with finite $p$-norm can be
rescaled to give a weakly $\theta$-summable $\Dd$.
Hence, as far as the rate of divergence is concerned,
there are plenty of sequences that give rise to 
a weakly $\theta$-summable $\Dd$.
In particular, when the system of graphs is the system of $d$-lattices,
so the edges corresponding to the $k$-th successive refinement have
weight $2^{dk}$. Then the sequence $\{\frac{1}{a_j}\}$ behaves
like the harmonic sequence,
as the $k$-th refinement adds about $2^{d(k-1)}$ edges
to the previous graph that has edges carrying weight $2^{d(k-1)}$.
Thus, $\{\frac{1}{a_j}\}$ has finite $p$-norm for any $p>1$.
And we can scale $\{\frac{1}{a_j}\}$ by an overall constant so that
its $p$-norm is less than 
$\sqrt{ \frac{2eu}{p}}
\left\lVert 
\left\lvert D_{\ker D^\perp } \right\rvert
^{-1}
\right\rVert_p ^{-1} $.

\bibliographystyle{plain}

\bibliography{thesis}
 
\end{document}